\documentclass{article}

\usepackage{microtype}
\usepackage{graphicx}
\usepackage{subcaption}
\usepackage{booktabs} 
\usepackage[breakable, theorems, skins]{tcolorbox}

\usepackage{rotating}

\usepackage{hyperref}
\usepackage[accepted]{icml2026}

\usepackage{amsmath}
\usepackage{amssymb}
\usepackage{mathtools}
\usepackage{amsthm}

\usepackage[capitalize,noabbrev]{cleveref}

\theoremstyle{plain}
\newtheorem{theorem}{Theorem}[section]
\newtheorem{proposition}[theorem]{Proposition}
\newtheorem{lemma}[theorem]{Lemma}

\theoremstyle{definition}

\theoremstyle{remark}

\usepackage{multirow}
\newtheorem{hypothesis}[theorem]{Hypothesis}

\usepackage{thm-restate}

\usepackage{tcolorbox}
\newcommand{\argmax}{\operatorname*{arg\,max}}

\usepackage{color,soul}

\usepackage[textsize=tiny]{todonotes}

\usepackage{graphicx}
\usepackage{subcaption}

\DeclareRobustCommand{\mybox}[2][gray!20]{%
\begin{tcolorbox}[  
        breakable,
        left=0pt,
        right=0pt,
        top=0pt,
        bottom=0pt,
        colback=#1,
        colframe=#1,
        width=\dimexpr\linewidth\relax, 
        enlarge left by=0mm,
        boxsep=5pt,
        arc=0pt,outer arc=0pt,
        ]
        #2
\end{tcolorbox}
}

\icmltitlerunning{Multi-Agent Systems are Mixtures of Experts}

\begin{document}

\twocolumn[
  \icmltitle{Multi-Agent Systems are Mixtures of Experts: \\ Who Becomes an Influencer?}

  \icmlsetsymbol{equal}{*}

  \begin{icmlauthorlist}
    \icmlauthor{Franka Bause}{equal,cispa}
    \icmlauthor{Jonas Niederle}{equal,cispa}
    \icmlauthor{Martin Pawelczyk}{vie}
    \icmlauthor{Rebekka Burkholz}{cispa}
  \end{icmlauthorlist}

  \icmlaffiliation{cispa}{CISPA Helmholtz Center for Information Security, Saarbrücken, Germany}
  \icmlaffiliation{vie}{Faculty of Computer Science, University of Vienna, Vienna, Austria}

  \icmlcorrespondingauthor{Franka Bause}{franka.bause@cispa.de}
  \icmlcorrespondingauthor{Jonas Niederle}{jonas.niederle@cispa.com}

  \icmlkeywords{Machine Learning, ICML}

  \vskip 0.3in
]

\printAffiliationsAndNotice{\icmlEqualContribution}

\begin{abstract}

The effectiveness of multi-agent LLM deliberation depends not only on the agents' individual predictions, but also on how they communicate and collaborate. We study this mechanism through the lens of Friedkin-Johnsen (FJ) opinion dynamics, a tractable model for analyzing stubbornness, influence, and opinion change in multi-agent systems that captures empirically observed deliberation patterns. We show that the FJ parameters are input-dependent, turning multi-agent deliberation into a mixture of experts. This perspective implies that multi-agent systems can outperform single agents and static ensembles when routing reflects agent competence. Since competence is latent in practice, we analyze how influence is established through observable proxies: agents' self-assessed confidence, their perceived confidence, and initial alignment with other agents' views.

\end{abstract}

\section{Introduction}

Large language models (LLMs) have led to significant advancements in natural language processing across various tasks. Recently, multi-agent systems (MASs) composed of interacting LLMs have attracted attention for their potential to improve performance, particularly in tasks involving strategic reasoning, negotiation, and generative design~\citep{du2024improving, wu2024autogen, qian2024chatdev,chen2024survey,hong2024metagpt}. In MASs, multiple agents communicate iteratively to deliberate and refine predictions, with the promise that diverse agents can contribute complementary expertise and improve decision-making.

However, the empirical benefits of MASs over single-agent models or static ensembles are mixed~\citep{smit2024mad, zhang2025stop, tran2026single}. A key factor in the success of MASs is understanding how influence is distributed during deliberation. Not all agents are equally persuasive, and the central question is:

\emph{What makes some agents more influential than others in a multi-agent deliberation?} 

We analyze this problem through the Friedkin-Johnsen (FJ) model of opinion dynamics~\citep{friedkin1990social}, which has been successfully applied to model belief propagation in social networks and, more recently, in LLM deliberations~\citep{donttrust}. 
While FJ dynamics have been used to assess security risks in MASs, we focus on the observation that the FJ model’s parameters are input-dependent, which leads to an important insight: multi-agent deliberation can be interpreted as a mixture of experts (MoE) system \citep{jacobs1991adaptive}. In this framework, each agent contributes to the final decision based on its initial belief and its influence weight, which varies depending on the input.

Accordingly, the MAS implicitly implements an adaptive routing mechanism. 
Its performance hinges on the degree to which influence is directed toward the most competent agents for a given input. Since true competence is latent, we focus on observable proxies for competence, such as self-assessed confidence, peer influence, and initial alignment with the other agents. Our findings demonstrate that relative confidence is the primary factor influencing the routing, with the initial alignment and social behavior also playing relevant roles.

\textbf{Contributions.} 
1) We use the Friedkin-Johnsen model to provide a tractable description of belief propagation, stubbornness, and peer influence in LLM-based MASs. 
2) We show that input-dependent FJ parameters induce a mixture-of-experts interpretation of multi-agent deliberation. 
3) We theoretically and empirically analyze how influence emerges from observable proxies for latent competence, identifying confidence and opinion alignment as central drivers of the implicit routing mechanism. 
Together, these results suggest that successful MAS design requires not only diverse agents, but also reliable mechanisms for routing influence toward agents that are locally competent.

\subsection{Related work}

\textbf{LLM-based MASs and deliberation.}
Recent work explores the use of LLMs interacting through structured or unstructured communication to improve reasoning and robustness \citep{wu2024autogen,li2023camel}. Techniques such as self-consistency~\cite{wang2023selfconsistency} and chain-of-thought prompting~\cite{DBLP:conf/nips/Wei0SBIXCLZ22} demonstrate that aggregating diverse opinions and reasoning can outperform single-pass inference. 
However, existing methods largely lack a principled framework on how opinions evolve through communication and how influence shapes the outcome. 

\textbf{Opinion dynamics and social influence models.} 
The Friedkin-Johnsen (FJ) model~\citep{deGroot_1974,friedkin1990social,friedkin2011social} %
has been applied extensively in the social sciences to understand how group consensus is formed, how stubbornness and belief retention interact with peer influence, and how changes in network structure can affect collective decision-making~\citep{parsegov2017novel,tian2018opinion,donttrust}.
Recently, it has been used to model opinions co-evolving with platform learning systems~\citep{wu2026opiniondynamicslearningsystems} and creating simulation environments of human opinion formation \citep{he2026opinion, openreview2025echo}.
Recent work has demonstrated that FJ dynamics can effectively capture the belief formation process in MASs to study the systemic risk induced by stubborn agents~\citep{donttrust}.
Surprisingly, compared to more complex cascades processes
\citep{burkholz2018fc,burkholz2018explicit,burkholz2018framework,burkholz_cascade_2020,gotovos2021scaling} of social influence~\citep{burkholz2016damage,flache2017models},
simple linear dynamics suffice to capture LLM deliberation.

\textbf{Interpretability of MASs.}
Research on the interpretability of MASs has explored how agents’ behavior can be explained, particularly in the context of collaborative problem-solving and collective decision-making~\citep{wooldridge2018multiagent,rosenfeld2019explainability,lee2025towards}. 
A challenge in multi-agent settings is understanding the emergent behaviors that arise from complex inter-agent interactions.
Our work adds to this growing body of research by providing a framework for explaining influence in MASs through the lens of Friedkin-Johnsen belief propagation \citep{donttrust} and a mixture-of-experts interpretation. 

\textbf{Mixture of experts.}
The concept of mixture-of-experts (MoE) models has been widely studied in the machine learning literature, where multiple models (or experts) are assigned responsibility for different regions of the input space. In these models, a gating network determines which expert(s) to trust based on the input, effectively routing the problem to the most relevant expert. MoE models have been used to improve performance on complex tasks by allowing specialized models to handle different subproblems~\citep{jacobs1991adaptive, jordan1994hierarchical,chen2022towards,sun2026robustnessmixturesexpertsfeature}.
In contrast to the standard MoE setting, we focus on collaborative MAS, where routing is not explicitly learned in a supervised setting but implicitly realized through deliberation.

\begin{figure*}
    \centering
    \begin{subfigure}{0.7\textwidth}
        \includegraphics[width=\textwidth]{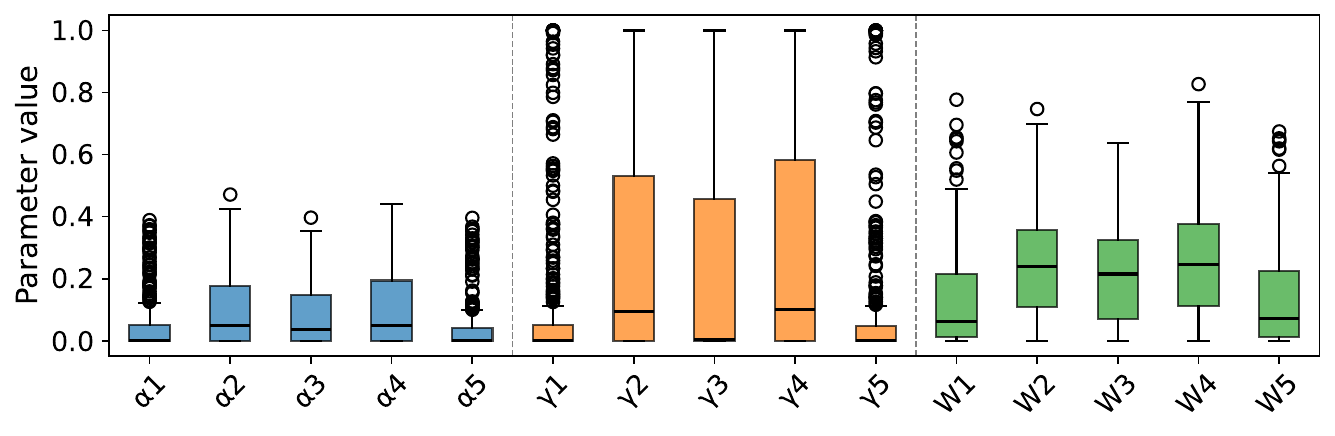}
        \subcaption{Parameter variability. W: weighted in-degrees (avg. over senders).}
        \label{fig:mmlu_pro_gpt_r_var}
    \end{subfigure}
    \begin{subfigure}{0.25\textwidth}
        \includegraphics[width=\textwidth]{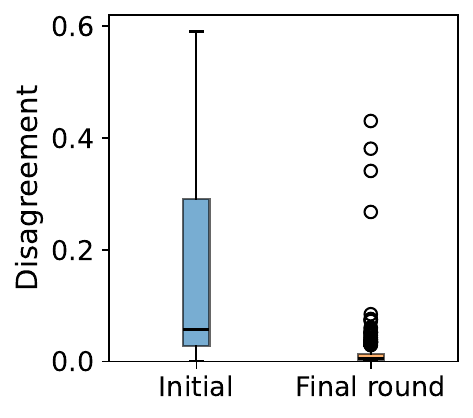}
        \subcaption{Belief disagreement.}
        \label{fig:concensus_box}
    \end{subfigure}
    \caption{FJ induces MoE. (a) FJ parameter variability across MMLU-Pro questions for GPT-5.4 Mini. (b) Despite diverse initial opinions (left), consensus is often reached (right).}
    \label{fig:concensus}
\end{figure*}

\section{Theoretical framework: Agentic deliberation as mixtures of experts}
\label{sec:theory}
Our work builds on the observation that LLM deliberation dynamics can be described by the Friedkin-Johnsen model (see also Table~\ref{tab:model_fit}), which enables us to cast multi-agent LLM systems as mixture of experts.
Accordingly, influence emerges as the result of a routing mechanism, which we set out to understand based on the confidence of agents and communication behavior.

\textbf{Friedkin-Johnsen (FJ) model.}
Let each agent $i \in V$ in an agentic network $G=(V,E)$ with $n=|V|$ agents hold a belief $b_i(t) \in \Delta^d$, where $\Delta^d \subset [0,1]^d$ is the $d$-dimensional simplex representing a probability distribution over potential outcomes, e.g., answers to a multiple-choice question.
Each agent $i$ is characterized by an innate belief $b_i(0)= s_i \in \Delta^d$, which may be interpreted as the agent's individual prediction before deliberation. 

The belief update at time $t+1$ in the \emph{Friedkin-Johnsen} model~\cite{friedkin1990social} for belief dynamics is defined as~\cite{donttrust}:
\begin{equation}
\label{eq:dyn}
\begin{split}
b_i(t+1) = & \underbrace{\gamma_i s_i}_{\text{Prior Belief Pull}} + \underbrace{(1-\gamma_i) \alpha_i b_i(t)}_{\text{Belief Retention}}  
\\ & + \underbrace{(1-\gamma_i) (1-\alpha_i) \sum_{j \in \mathcal{N}i} w_{ij} b_j(t)}_{\text{Peer Influence Pull}},
\end{split}
\end{equation}
where $\gamma_i \in [0,1]$ denotes the attachment to innate beliefs also called stubbornness, $\alpha_i \in [0,1]$ represents the weight given to the previous state, and $W = [w_{ij}]$ is a row-stochastic influence matrix where $\sum_j w_{ij} = 1$ (and $w_{ii}=0$, as self-loops are covered by $\alpha_i$). The term $(1-\gamma_i)(1-\alpha_i)$ represents the agent's susceptibility to external influence.  
We can also formulate this in matrix notation. Let $B(t)\in\mathbb R^{n\times d}$ denote the matrix whose $i$-th row is $b_i(t)^\top$, and let $S\in\mathbb R^{n\times d}$ denote the matrix whose $i$-th row is $s_i^\top$. Define diagonal matrices $\Gamma=\operatorname{diag}(\gamma_1,\dots,\gamma_n)$, $A=\operatorname{diag}(\alpha_1,\dots,\alpha_n)$,
and $H=(I-\Gamma)\bigl(A+(I-A)W\bigr)$. Then the deliberation dynamics can be written compactly as $B(t+1)=\Gamma S + H B(t)$.

As LLM deliberation dynamics are captured by the FJ model, the multi-agent system (MAS) computes a convex combination of the initial beliefs of its agents.
\begin{proposition}[Graph-induced convex combination \citep{friedkin1990social,proskurnikova2017tutorial}]
\label{prop:fj-ensemble}
Assume that $\rho(H)<1$, where $\rho(\cdot)$ denotes the spectral radius. Then the FJ dynamics converge to the unique equilibrium $B^\star=(I-H)^{-1}\Gamma S$. $M := (I-H)^{-1}\Gamma = [m_{ij}]$
is nonnegative and row-stochastic. Consequently, each equilibrium belief is a convex combination of the innate beliefs: $   b_i^\star = \sum_{j=1}^n m_{ij}s_j$. 
If the agentic network belief $b_{\mathrm{out}}^\star = \sum_{i=1}^n \eta_i b_i^\star = \sum_{j=1}^n \pi_j s_j$ with $\pi^\top=\eta^\top M$ is formed by a linear aggregation of all agents' beliefs (typically $\eta_i=1/N$), the agentic network is equivalent to an ensemble with weights $\pi_j$ that depend on the FJ parameters $(\Gamma, A, W)$.
\end{proposition}
A system that follows Friedkin-Johnsen dynamics with static parameters $(\Gamma, A, W)$ would therefore implement an ensemble.
This can be already beneficial compared to a single-agent model, as ensembles can combine the opinion of diverse agents, achieve a variance reduction and are potentially more robust to distribution shift \citep{dietterich2000ensemble,reddy2026when,reddy2026boosting}.
A direct implication of this insight is that agentic networks tend to benefit from high agent diversity.

Yet, our empirical findings suggest that MASs can be even more powerful, as they do not rely on static FJ parameters, but adapt them to the input questions $x$, resulting in FJ parameters $(\Gamma(x), A(x), W(x))$ that induce input-dependent routing of agent influences $\pi_i(x)$.
\begin{hypothesis}\label{hypo:moe}
A multi-agent LLM system (MAS) can be cast as mixture of experts (MoEs)  $b^\star(x) = \sum_{j=1}^n \pi_j(x) s_j(x)$, where the router $\pi_j(x)$ depends on the input $x$. 
\end{hypothesis}
 Fig.~\ref{fig:mmlu_pro_gpt_r_var} provides empirical evidence for this hypothesis.
It implies that the two central quantities of interest that determine the performance of the MAS are a) the diversity of agents and b) the routing.
While a) is a critical design choice, it is often given and only mildly influenced by prompting in our LLM deliberation experiments.
b) Routing is performed automatically by the LLM agents. 
The main purpose of the introduced FJ framework is to facilitate the discovery of its main underlying mechanisms.
Our analysis reveals that routing primarily relies on initial beliefs $s_i(x)$, even though these are not explicitly communicated by the LLM agents.
However, they serve as proxy for the agents' states that also informs their behavior. 
They define their confidence, competence, and to which degree their opinions align initially.

Based on this hypothesis, i.e.\ an MAS acts as an MoE that performs routing based on initial agent beliefs ($\pi = \pi(S)$), for which we provide empirical evidence, our following theoretical discussion aims to derive conditions when MASs outperform single agents and even ensembles. 

\subsection{When do MASs outperform single agents and ensembles?}
MASs have shown promising performance gains over single agents.
This is especially true when the task has compositional structure, agents contribute diverse information or reasoning strategies, and communication allows for error correction \citep{du2024improving,wu2024autogen,qian2024chatdev,hong2024metagpt}.
However, recent evaluations find only modest or inconsistent gains over strong single-agent baselines \citep{smit2024mad,zhang2025stop,cemri2025why,tran2026single}, calling for a deeper understanding of successful MAS design principles.

Our link to MoEs allows us to draw on theoretical insights.
MoEs are a special case of ensembles. 
For ensembles, the routing is constant across inputs. 
While it is well known that agent diversity is critical to their function \citep{dietterich2000ensemble}, theoretical studies of MoEs have focused on router properties \citep{jacobs1991adaptive,jordan1994hierarchical,jiang1999hierarchical} or how clustered structure in the input data can be exploited \citep{chen2022towards,sun2026robustnessmixturesexpertsfeature} in the context of supervised learning problems, where the router has access to predictive features of labels.

In our context of LLM deliberation, no router has been trained explicitly and the MAS does not have access to label information.
Intuitively, an ideal router would assign high influence to competent agents.
Yet, agent competence is input question dependent, latent, and therefore unavailable information.
Agent beliefs, however, can be interrogated.
We therefore ask the question: \emph{When can belief-derived signals identify which agent is locally competent and drive MAS performance?} 

To answer this question, our analysis decomposes performance into three terms: global ensemble diversity, local competence diversity, and routing regret. 
This lets us compare single agents, static ensembles, and MoE-style deliberation in terms of agent diversity, complementary information, and confidence-based routing that rests on sufficiently well-calibrated agents, whose confidence reflects competence. 
In particular, we show that adaptive deliberation improves over fixed ensembling only when the gain from local specialization exceeds the diversity lost by no longer averaging all agents and the regret incurred by imperfect confidence-based routing.

\textbf{Belief-dependent routing, diversity, and performance.}
To formalize this intuition, we assume that each agent $j\in[n]$ outputs an initial belief $s_j(X)\in\Delta^d$ about the right answer to question $X$, which is a random variable over the set of possible questions, and write $S(X)=(s_1(X),\dots,s_n(X))$ for the collection of all beliefs.
A MAS outputs $f(X) = \sum_{j=1}^n \pi_j(S(X))s_j(X)$ and its performance is measured with the Brier loss $\ell(y,p)=\|p-e_y\|_2^2$, where $e_y$ denotes the one-hot vector for belief in the correct answer $y$.
For each agent, define the belief-conditional risk $r_j(S)= \mathbb E \left[\|s_j(X)-e_Y\|_2^2\mid S(X)=S\right]$.
This is the true local competence of agent $j$ given the observable beliefs $S$, which is generally unknown to the system.
For any weight vector $a\in\Delta^n$, define the mixture $\bar s_a =\sum_{j=1}^n a_j s_j$ and the local initial belief diversity, which measures how diverse the agents' beliefs are under weights~$a$: 
\begin{align}\label{eq:diversity}
\hspace{-0.2cm} D_a(S) = \sum_{j=1}^n a_j\|s_j-\bar s_a\|_2^2 = \frac12 \sum_{i,j=1}^n a_i a_j \|s_i-s_j\|_2^2.   
\end{align}

\begin{lemma}[Local ambiguity decomposition]
\label{lem:local-ambiguity}
For any belief-dependent weights $a(S)\in\Delta^n$,     
\begin{align*}
\mathbb E
    \left[
        \left\|
            \sum_{j=1}^n a_j(S)s_j(X)-e_Y
        \right\|_2^2
        \mid S
    \right] \\
    =
    \sum_{j=1}^n a_j(S) r_j(S)
    -
    D_{a(S)}(S).
    \end{align*}
\end{lemma}

This composition has the following interpretation: 
The first term rewards putting weight on locally competent agents. 
The second term rewards averaging diverse beliefs.

\textbf{Comparing MAS, single agents, and ensembles.}
This composition enables us to compare different routing choices $a(S)\in\Delta^n$ over the set of all questions by taking an average also over $X$ and $S$ additionally to $Y$, yielding $L = \sum_{j=1}^n  \mathbb E [a_j(S) r_j(S)] -  \mathbb E [D_{a(S)}(S)]$.
In the following, we use the notation $a(S)$ for a general router, $\pi(S)$ for a MoE, and $\Psi(S)=\Psi$ for an ensemble, where the routing is independent of the input. 

\begin{theorem}[MAS vs. single agent]\label{th:mas_vs_single}
A multi-agent system with the mixture of experts routing $\pi(S)$ outperforms the best single agent $j^*$ with lowest risk $\min_{j\in[n]}\mathbb E[r_j(S)]$ if
\begin{align}\label{eq:MASvsSingle}
\hspace{-0.2cm}\underbrace{\mathbb E [r_{j^*}(S)- \min_j r_j(S) ]}_{\text{specialization gain}} + \underbrace{\mathbb E\left[D_{\pi(S)}(S)\right]}_{\text{local diversity}}
    >    \underbrace{\mathbb E[\delta_{\pi}(S)]}_{\text{routing regret}},    
\end{align}
where $\delta_{\pi}(S)=\sum_{j=1}^n \pi_j(S) r_j(S) -\min_{j\in[n]} r_j(S)$ defines the routing regret.
\end{theorem}
The theorem is verified experimentally by Fig.~\ref{fig:app:confusion_matrices_gpt}.
The proof follows directly from taking the average over the respective losses and demanding the loss of the MoE to be lower. 
The first term is the specialization advantage over the best global agent. 
It is large when no single agent is best everywhere. 
The second term is the local diversity benefit retained by the MoE. 
The right-hand side is the cost of imperfect routing.~Therefore,
\mybox{
\centering
MoE beats the best single agent if specialization gain plus local diversity exceeds routing regret.}

This implies that a multi-agent system does not improve simply by combining many agents. Its agents have to be locally competent, complementary, and the routing mechanism can identify them from observable belief signals.

However, note that this analysis does not exclude the existence of another more powerful and potentially more complex single expert model, which could e.g. be obtained by distillation of the MAS.
It could also implement a similar mechanism as routing by choosing good  answers from a sample like self-consistency (SC) \citep{zhang2025stop}.
At the same time, such single models could potentially also be valuable experts in a MAS if they are combined with other complementary experts.

Note that ensembles are special cases of MoEs and therefore follow a similar logic. 
However, they cannot achieve a local specialization gain through routing and the above condition simplifies to  $\mathbb E[D_a(S)] > \sum_{j=1}^n a_j\mathbb E[r_j(S)]-\min_{j\in[n]}\mathbb E[r_j(S)]$, highlighting the relevance of expert diversity.
The difference in this analysis points also out how MoEs can potentially outperform ensembles, i.e. by exploiting local diversity.
\begin{theorem}[MAS vs. ensemble]
A multi-agent system with the mixture of experts routing $\pi(S)$ outperforms an ensemble with constant routing weights $\psi$ if
\begin{align}\label{eq:MASvsEnsemble}
    \underbrace{\mathbb E \left[\sum_{j=1}^n\bigl(\psi_j-\pi_j(S)\bigr) r_j(S)
    \right]}_{\text{local routing gain}} > \underbrace{\mathbb E \left[ D_\psi(S)-D_{\pi(S)}(S) \right]}_{\text{diversity loss}}.
\end{align}
\end{theorem}
The left-hand side is the gain (resulting from lower average risk) from assigning more weight to locally competent agents. 
The right-hand side is the diversity that could be lost by moving away from the fixed ensemble. Thus, 
\mybox{
\centering
MoE wins if local competence routing gain exceeds lost ensemble diversity.}

\begin{figure}
  \centering
  \vspace{-5pt}
  \includegraphics[width=\columnwidth]{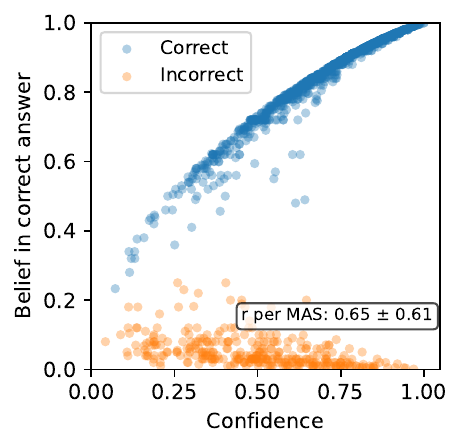}
  \caption{Association of competence and confidence for ChatGPT-5.4 Mini on MMLU-Pro.
  }
  \label{fig:comp_conf}
\end{figure}
\textbf{Hard routing towards the most competent agent.}
As a special case, we can also consider hard routing, where a single agent $j'(S)$ receives all weight, leading to the condition $\mathbb E
    \left[
        \sum_{j=1}^n
        \psi_j r_j(S)
        -
        r_{ j'(S)}(S)
    \right]
    >
    \mathbb E[D_\psi(S)]$,
because the local diversity $D_{\pi(S)}(S)=0$.
Thus, hard routing must compensate for the diversity benefit with a sufficiently better local expert.
Our analysis of real LLM deliberation reveals a preference of MAS to select a single or only a few experts, suggesting a strong tendency towards local specialization rather than exploiting diversity through ensembling. %
Interestingly, even though the agents do not appear to be very diverse with respect to their average performance, their intrinsic randomness leads to local diversity in the their initial beliefs, which is exploited by the routing. 
In the following, we discuss scenarios how such routing is enabled by proxies for local expert competence in the initial beliefs.

\subsection{Confidence-based routing as a competence proxy}
\begin{figure*}
    \centering
    \includegraphics[width=\linewidth]{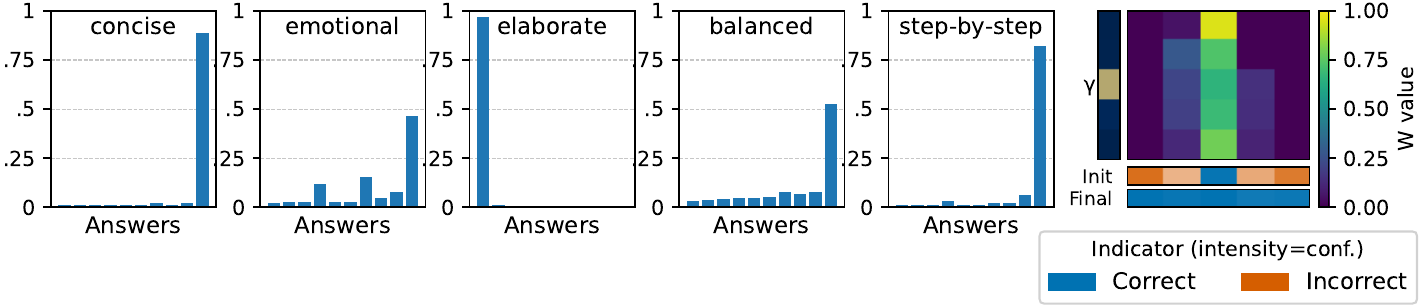}
    \caption{Example with communication styles on \emph{GPT-5.4 Mini}: Initial beliefs of agents (left). FJ weight matrix, $\gamma$, and color coded belief in correct answer (right). The most confident agent convinces the majority to change their answer.}
    \label{fig:sample_290}
\end{figure*}

The local risks $r_j(S)$ are not observed during deliberation. 
The router must therefore use proxies derived from the beliefs.
A natural self-assessment of competence is their confidence $C_j(S) \in [0,1]$ in a specific answer, which we propose to measure based on the entropy $\mathcal H(s_j)$ of their initial belief:
\begin{align}\label{eq:confidence}
       C_j(S) = 1- \frac{1}{\log d} \mathcal H(s_j) =1+\frac{1}{\log d} \sum_{c=1}^d s_{j,c}\log s_{j,c}.
\end{align}
A confidence-based router may take the form $\pi_j(S)
    = \exp(\beta C_j(S))/ \sum_{\ell=1}^n \exp(\beta C_\ell(S))$, 
where $\beta\ge 0$ controls how strongly the router favors confident agents.
Fig.~\ref{fig:sample_290} presents an example, where the MAS simply picks the opinion of the initially most confident agent.
As we show, not only absolute confidence of agents matters but also confidence relative to the other agents.
Thus, all predictions are taken into account. 

Such routers are beneficial only when confidence is sufficiently calibrated with competence (see also Fig.~\ref{fig:comp_conf}). 
One way to formalize calibration is to assume that there exists a decreasing function $\phi$ such that $r_j(S)\approx \phi(C_j(S))$.
Equivalently, agents with higher confidence should have lower conditional risk.
Let $j_C(S)\in\arg\max_j C_j(S)$ be the most confident agent, and let $j^\star(S)\in\arg\min_j r_j(S)$ be the most competent agent. 
For hard confidence routing, the routing regret is $\delta_C(S) = r_{j_C(S)}(S)-r_{j^\star(S)}(S)$.
The confidence-routed MoE beats a fixed ensemble with weights $\psi$ if
 $   \mathbb E[G_\psi(S)]
    >
    \mathbb E[\delta_C(S)]
    +
    \mathbb E[D_\psi(S)],
$
where $G_\psi(S) = G_w(S)  = \sum_{j=1}^n \psi_j r_j(S) - \min_{j\in[n]} r_j(S)$ measures how much a fixed ensemble wastes probability mass on agents that are not locally optimal. 
This condition captures both the strength and the limitation of confidence routing. 
If confidence reliably identifies competence, then $\delta_C(S)$ is small and the MoE can exploit local specialization. 
If agents are overconfident when wrong, then $\delta_C(S)$ can be large and the MoE may underperform a fixed ensemble.
To get an intuitive understanding of this, let us discuss two examples.

\textbf{Case: Mutually exclusive and optimally calibrated agents.}
Ideal agents of a MAS have a specialized local expertise that is complementary to other agents and signal reliably when they are competent.
Optimal agents $j$ with limited performance budget $B_j$ would specialize and thus spend their budget on disjoint sets of questions.

Let the input space decompose into disjoint regions,
$
    \mathcal X=\bigcup_{j=1}^n \mathcal X_j$,
$ \mathcal X_i\cap \mathcal X_j=\emptyset$
for $i\neq j$ with
$
    \Pr(X\in \mathcal X_j)=\rho_j$,
    $\rho_j>0$, 
    $\sum_{j=1}^n \rho_j=1$. 
On region $\mathcal X_j$, agent $j$ is competent and confident, while all other agents are uninformative. 
Let $p=1-\varepsilon$ with $\varepsilon\in(0,1-1/d)$, and let
$
    u=1/d
$
denote the probability assigned to the true class by the uniform distribution.
For every input $x\in\mathcal X_j$, assume that agent $j$ assigns probability $p$ to the true class, while every other agent assigns probability $u$ to the true class:
$s_{j,Y}(x)=p$, $s_{i,Y}(x)=u$ for all $i\neq j$. 
Thus, each agent is competent on exactly one region and uninformative elsewhere.
Assuming a symmetric belief that assigns probability $q$ to the correct label and distributes the remaining mass uniformly over the incorrect labels, achieves a Brier loss of $\mathcal{B}_{d}(q)  = (1-q)^2 + (d-1)\left(\frac{1-q}{d-1}\right)^2
= \frac{d}{d-1}(1-q)^2$.
A MoE router that selects the competent agent on every region thus achieves Brier loss
$\mathcal{L}_{\text{MoE}} = \mathcal{B}_{d}(p)$.
In comparison, a fixed ensemble with weights $\psi\in\Delta_n$ assigns on region $\mathcal{X}_j$ probability $q_j(\psi) = u+(p-u)\psi_j $ to the correct label. Its expected Brier loss is therefore
$ L_{\mathrm{ens}}(\psi) = \sum_{j=1}^n \rho_j \mathcal{B}_{d} \left(u+(p-u)\psi_j\right)$. 
\begin{proposition}[MoE advantage under mutually exclusive competence]
\label{prop:n-agent-moe-advantage}
Assume $n\ge 2$, $\rho_j>0$ for all $j$, and $p>u$. Then $ L_{\mathrm{MoE}}    < \inf_{\psi\in\Delta^N} L_{\mathrm{ens}}(\psi)$.
\end{proposition}
For balanced regions $\rho_j=1/n$ for all $j$, symmetry and strict convexity imply that the optimal fixed ensemble is the uniform ensemble $ \psi_j^\star=1/n$.
Writing $q_n=u+\frac{p-u}{n}$, the loss gap then becomes $L_{\mathrm{ens}}^\star-L_{\mathrm{MoE}} = \frac{d}{d-1} \left[
(1-q_n)^2-(1-p)^2 \right] > 0$.
For fixed $p$ and $d$, this gap increases with the number of agents. 
The best single agent $j^{\mathrm{glob}} \in \arg\max_{j\in[n]}\rho_j$ is the agent whose competence region has largest probability.
Its loss is $L_{\mathrm{single}}^\star = 
\rho_{j^{\mathrm{glob}}}\mathcal{B}_{d}(p) + \bigl(1-\rho_{j^{\mathrm{glob}}}\bigr)\mathcal{B}_{d}(u)$.
Consequently, $L_{\mathrm{single}}^\star-L_{\mathrm{MoE}} = \bigl(1-\rho_{j^{\mathrm{glob}}}\bigr)
\left[ \mathcal{B}_{d}(u)-\mathcal{B}_{d}(p) \right]$. This means, the MoE therefore strictly improves over the best single agent whenever at least two competence regions have positive~probability.

\textbf{Effect of routing errors.}
The preceding result assumes that confidence perfectly identifies the competent agent. 
Suppose instead that the router selects the correct agent with probability $1-\delta$ and selects an uninformative agent with probability $\delta$. 
Then the routed system assigns true-class probability $p$ with probability $1-\delta$ and $u$ with probability $\delta$, 
Its expected Brier loss is $L_{\mathrm{router}}(\delta)
= (1-\delta)\mathcal{B}_d(p) + \delta\mathcal{B}_d(u)$.
In the balanced setting, this router improves over the optimal fixed ensemble if and only if $\delta
<
\frac{
\mathcal{B}_d(q_n)-\mathcal{B}_d(p)
}{
\mathcal{B}_d(u)-\mathcal{B}_d(p)
}$.
Thus, confidence-based routing is beneficial only when the confidence signal is sufficiently reliable. 
This formalizes the central limitation of task-dependent FJ deliberation: adaptive influence can realize MoE-like gains, but miscalibrated confidence can route trust to the wrong agent and erase the advantage.

As experts are mutually exclusive in this ideal case, a MAS can generally route towards the single most competent expert, while the other agents focus their performance budget on other questions.
Note that the maximum belief predictions of ensembles are still correct (albeit with lower certainty).
As a consequence, expert optimization could make complex MAS interaction unnecessary. 
The real advantage of MoE routing is realized with imperfect agents, which are more common in practice.

\textbf{Case: Routing with imperfect agents.}
More interesting cases that highlight the advantages of MoE routing  assume imperfect agents.
Let us consider a scenario that is comparable with the observed case in Fig.~\ref{fig:sample_290}, where the majority of agents makes a wrong prediction, confusing an ensemble, but the most confident agents convinces the rest during deliberation. 
As before, we observe disjoint regions of competence but with less well calibrated agents that correlate their belief in wrong labels, which makes the problem harder to solve.

Assume $d\geq 3$. On region $\mathcal{X}_k$, let $y_k$ be the correct label and let $z_k\neq y_k$ be a common incorrect label. The competent agent $k$ predicts $s_{k,y_k}=p$, $s_{k,c}=b:=\frac{1-p}{d-1}$ for $c\neq y_k$.
Every non-competent agent predicts $s_{j,y_k}=u$, $s_{j,z_k}=c$, $s_{j,\ell}= r:=\frac{1-u-c}{d-2}$ for $j\neq k$ and $\ell\notin\{y_k,z_k\}$. Assume $0<u<\frac{1}{d}<p<1$, $0<c<1-u$, $c>\max\{u,r\}$.
The last condition formalizes that the non-competent agents concentrate their probability on the same wrong answer. We further assume that the competent agent is the most confident:
$C_k(S(x))>C_j(S(x))$ for every $j\neq k$ and $x\in\mathcal{X}_k$. 
Hard confidence routing selects the competent agent and achieves
$L_{\mathrm{MoE}} = \mathcal{B}_d(p)$ as before.

For a fixed ensemble $\psi\in\Delta_n$, its prediction on region $\mathcal{X}_k$ has probabilities $q_{y,k} =
u+(p-u)\psi_k$, $q_{z,k}= c+(b-c)\psi_k$ and $q_{o,k} = r+(b-r)\psi_k$ on each of the remaining $d-2$ labels. For balanced regions, its expected Brier loss is $L_{\mathrm{ens}}(\psi) = \frac{1}{n} \sum_{k=1}^n
\left[(1-q_{y,k})^2 + q_{z,k}^2 + (d-2)q_{o,k}^2\right]$. 

\begin{proposition}[Confidence routing beats fixed ensembling]
\label{prop:n-agent-d-label-confidence-routing}
Assume $n\geq 2$, $d\geq 3$, and the conditions above. Then the optimal fixed ensemble is the uniform ensemble with $\psi_k^\star=\frac{1}{n}$ and confidence routing strictly improves over it:
$L_{\mathrm{MoE}} <L_{\mathrm{ens}}(\psi^\star)$. 
More precisely, with $t=\frac{n-1}{n}$, the loss gap is
\begin{align*}
L_{\mathrm{ens}}(\psi^\star) - L_{\mathrm{MoE}}
={}& 2t(p-u)(1-p+b) \\
&+ t^2 \big[\, (p-u)^2 + (c-b)^2 \\
&\hphantom{{}+ t^2 \big[\,} + (d-2)(r-b)^2 \,\big] > 0.
\end{align*}
\end{proposition}

Under the uniform ensemble, the correct and common incorrect labels receive probabilities $\bar q_y = \frac{p+(n-1)u}{n}$, $\bar q_z = \frac{b+(n-1)c}{n}$. 
Because $c>r$, the common incorrect label is the most likely incorrect label. The uniform ensemble
therefore predicts the wrong answer whenever
$
b+(n-1)c
>
p+(n-1)u.
$
By contrast, hard confidence routing predicts correctly on every region because $p>1/d$ implies $p>b$.

\textbf{Strengths and limitations.}
This theoretical perspective clarifies both the strengths and the limitations of FJ-based multi-agent deliberation. 
Its strength is that deliberation can implement a task-adaptive mixture of agents: when the FJ parameters track competence, the system can route influence toward agents that are locally reliable. 
This enables gains from specialization that a fixed ensemble cannot realize.
Its limitation is that routing depends on imperfect competence signals. 
If confidence is miscalibrated, if behavioral confidence is uninformative, or if opinion alignment merely reflects correlated errors, then the induced influence weights may amplify the wrong agents. 
In particular, high agreement among agents can be beneficial when it reflects independent corroboration, but harmful when it reflects shared bias. 
Similarly, high confidence can be beneficial when calibrated, but harmful when associated with overconfident mistakes.

\section{Experiments}

The main goal of our experiments is to demonstrate the utility of the Friedkin-Johnsen modeling approach and mixture of expert interpretation for studying the emergence of influence in MASs. %

Our experimental evaluation considers the MMLU-Pro~\cite{wang2024mmlupro}, BBQ~\cite{DBLP:conf/acl/ParrishCNPPTHB22}
and CommonsenseQA (CSQA)~\cite{DBLP:conf/naacl/TalmorHLB19} dataset.
We adopt the setup and subset of $100$ questions of %
CSQA from~\cite{donttrust} and sample $300$ questions from MMLU-Pro and BBQ each (balanced for the categories). We conduct our experiments with three representative language models, \emph{GPT-5.4 Mini}, \emph{Qwen2.5-14B-Instruct}, and \emph{Qwen2.5-72B-Instruct-GPTQ-Int8}, enabling comparison across different model scales. We run our experiments over $3$ different seeds.

We use the same base system prompt for each agent and extend it to create diversity by prompting agents to assume different \emph{roles}~\cite{DBLP:conf/naacl/KongZCLQSZWD24}, like doctor, mathematician, or careless student, or answering in a specific \emph{communication style}, e.g. concise, balanced, or emotional. %
Prompting agents to act as \emph{experts} in supercategories~\cite{ko2026social} of the MMLU-Pro dataset and giving them access to specific example questions can also introduce agent diversity. The scenario \emph{neutral} refers to no diversifying prompts being added in the system prompt.
For the user prompts, communication between agents, and generating prompts for the next deliberation round, we use the setup of~\cite{donttrust}.
The prompts can be found in Appendix~\ref{app:prompts}. 
We use $5$ agents, communicating in a \emph{complete} graph structure over $5$~rounds.

\textbf{Competence.} We measure agent competence by the magnitude of their belief in the correct answer. This measure is unknown to the system but we can evaluate it based on known labels. To understand on what basis a MAS approximates agent competence, we define several other variables. 

\textbf{Confidence.} The \emph{confidence} of an agent $j$ is defined in Eq.~(\ref{eq:confidence}). The \emph{relative-confidence} of agent $j$ is defined as $R_j(S) = \frac{C_j(S)}{C_{(n-1)}(S)}$
where $C_{(n-1)}(S)$ is the second most confident agent in the system.

\textbf{Influence.} The \emph{influence} of agent $j$, normalized by the maximum agent influence, is defined as
$I_j(S) = \frac{\pi_j}{\max_i \, \pi_i}$, where $\pi_j$ is defined in Prop.~\ref{prop:fj-ensemble}.

\textbf{Peer Influence.} Let $B = A + (I-A)W$ and define $\tilde{B}$ as $B$ 
with zero diagonal entries, i.e., $\tilde{B}_{ij} = B_{ij}\,\mathbf{1}[i \neq j]$. Here, $A$ and $W$ refer to FJ matrices, as defined in Section \ref{sec:theory}.
The \emph{peer influence} of agent $j$ is $P_j = \frac{\pi^{\mathrm{peer}}_j}{\max_k\, \pi^{\mathrm{peer}}_k}$, where 
  $\pi^{\mathrm{peer}}_j = \sum_{i=1}^{n} \tilde{B}_{ij}$ and $\pi^{\mathrm{peer}}_j$ is the $j$-th column sum of $\tilde{B}$.

\textbf{Disagreement.} The \emph{disagreement} of a system $S$ is the mean distance of each agent's initial belief from the average system opinion: $  Di(S) = \frac{1}{n}\sum_{j=1}^{n} \left\|s_j - \frac{1}{n}\sum_{i=1}^n s_i \right\|_2.$

\textbf{Alignment.} The \emph{alignment} of an agent with the MAS is given by the cosine similarity of its initial belief to the mean opinion: $Al_j(S) = \frac{s_j \cdot \frac{1}{N} \sum_{i=1}^N s_i}{\|s_j\|_2 \cdot \|\frac{1}{N} \sum_{i=1}^N s_i \|_2}$

\textbf{Alignment score.} The \emph{alignment score} is a binary variable, describing whether an agents initial answer matches the groups initial answer: 

$As_j(S)=\begin{cases} 1, \text{ if }\argmax(s_j) = \argmax(\frac{1}{N} \sum_{i=1}^N s_i) \\ 0, \text{ else}
\end{cases}$

\textbf{Alignment count.} The \emph{alignment count} of an agent gives the number of other agents, that share the same initial answer: $Ac_j(S) = \sum_{i=1, i\neq j}^N \delta_{\argmax(s_j), \argmax(s_i)},$ where $\delta$ is the Kronecker delta.

\begin{figure}
  \centering
  \includegraphics[width=\columnwidth]{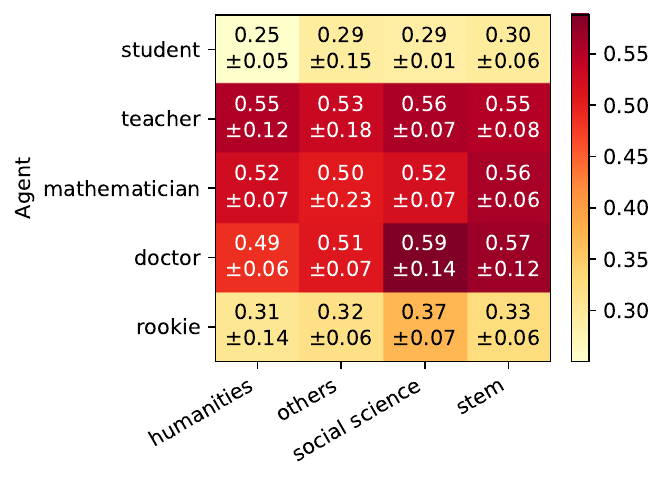}
  \caption{Influence (mean $\pm$ 95\% confidence interval) for different roles, differences in influence suggest \textit{perceived confidence} plays a role in FJ dynamics.
  }
  \label{fig:roles_heatmap}
\end{figure}

We empirically study the FJ dynamics in agentic systems and investigate how the routing weights are dependent on the different variables. If not indicated otherwise, the plots show results from the MMLU-pro dataset with role prompts with \emph{GPT-5.4 Mini}. Plots for the other configurations can be found in Appendix~\ref{app:addplots} and in the supplementary material.

\textbf{Agentic system often outperforms ensemble.}
We compare the performance of the agentic system to the baseline of taking the answer with maximum belief (averaging over the initial agent beliefs), and using a task independent ensemble by fitting constant FJ parameters over all samples, optimizing for predictive performance. Table~\ref{tab:global_ensemble_accuracy} in the appendix shows that agentic systems can perform better than both baseline and FJ ensemble, underlining the theory that the routing is input dependent, suggesting a MoE model.

\textbf{MASs implement MoEs.}
Table~\ref{tab:model_fit} shows that the FJ model fits the dynamics of the agentic collaboration well for all datasets, models and prompts styles. We can therefore leverage these fitted parameters to investigate the variability in agent weight between tasks. High variability of the FJ parameters over samples, as seen in Fig.~\ref{fig:mmlu_pro_gpt_r_var}, suggests that the parameters are input dependent. Consequently, the final weighting of the agents' initial beliefs, $\pi$ in Proposition~\ref{prop:fj-ensemble}, are also input dependent, empirically confirming \mbox{Hypothesis~\ref{hypo:moe}.} %

\textbf{Tendency towards consensus in MASs.} Despite the diversity in initial opinions (see Fig.~\ref{fig:concensus}), agents consistently reach consensus in the final round. This is unexpected under FJ dynamics and only happens under specific conditions, e.g. when little or no stubbornness is present. Furthermore, we observe that influence is concentrated on a small subset of agent, as can be seen in Fig.~\ref{fig:sample_selected} and in the individual samples in Fig.~\ref{fig:samples_heatmap} (see appendix).
Influence is strongly related to stubbornness ($\gamma$), as can be seen in Fig.~\ref{fig:pi_influence_gamma}: Within an agentic system, stubborn agents become the most influential. %

The fact that a small subset of agents dominate the final consensus and that influence is input dependent, shows that agentic systems make strong routing decisions, and motivates the need for understanding this routing behavior.

\textbf{Competence is associated with influence.}
Fig.~\ref{fig:infl_comp} in the appendix shows that there is a tendency for competent agents to get more influence, which is an indicator of effective routing. 
Particularly, confidence relative to other agents is a stronger predictor than absolute confidence, highlighting the social aspect of influence emergence.

Interestingly, the role of an agent can change its influence, suggesting that \emph{perceived confidence} also plays a role in the FJ dynamics (see Figs.~\ref{fig:roles_heatmap} and.~\ref{fig:mmlu_pro_gpt_r_classification}).

\begin{figure}
  \centering
  \includegraphics[width=\columnwidth]{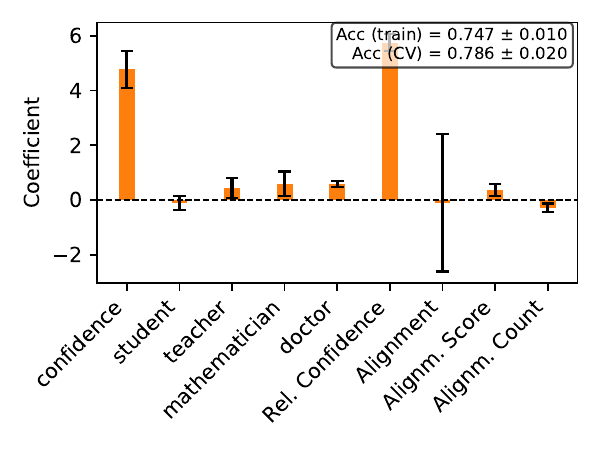}
  \caption{Logistic regression coefficients classifying the most influential agent. Confidence and competence are the strongest positive predictors. Role coefficients suggest perceived confidence plays a role in FJ dynamics.
  }
  \label{fig:mmlu_pro_gpt_r_classification}
\end{figure}

\textbf{Confidence leads to influence.}
To understand the routing behavior, we investigate the relation between metrics of the initial beliefs and competence with influence of agents by training random forests to regress influence to classify (Fig.~\ref{fig:regression_class}, appendix) if an agent becomes the most influential one. The models achieve a test $R^2$ of $0.7$ and accuracy of $0.9$ for regression and classification, respectively, which suggests that influence can largely be explained by the considered variables.
We also use logistic regression to gain insights into the relative impact of different variables. Fig.~\ref{fig:mmlu_pro_gpt_r_classification} shows that confidence both in absolute terms and relative to the second most confident agent is predictive of influence. Furthermore, competence, prompt style and alignment have an effect. Other datasets show similar behavior (see supplementary material).
The positive relation of confidence to influence as well as peer-influence can also be observed in
Figs.~\ref{fig:app:con_inf_gpt_neutral}-\ref{fig:app:con_inf_qwenl_neutral} in the appendix.

\textbf{Confident agents are more stubborn.} Naturally, alignment with the majority opinion plays in a role in the rate to which an agent sticks to its opinion. Agents that align with the majority stick to their opinion more frequently. However, both for aligned and misaligned agents, confident agents tend to stick more to their initial belief, as can be seen in Fig.~\ref{fig:app:change_rate_neutral} in the appendix.

\begin{table} 
\centering
  \caption{Model fit measured in KL divergence and MSE aggregated across datasets, prompt types, seeds and models.}
\label{tab:model_fit}
\begin{tabular}{lc}
\toprule
Metric & Mean $\pm$ 95\% CI \\
\midrule
KL Divergence & $0.0470 \pm 0.0034$ \\
MSE & $0.00198 \pm 0.00026$ \\
\bottomrule
\end{tabular}
\vspace{-0.25cm}
\end{table}

As noted before, there is a strong correlation between influence and stubbornness (see Fig.~\ref{fig:pi_influence_gamma}), suggesting that the routing behavior emerging in MASs is shaped also by an agents' resistance to changing their opinion. While we see that confident agents tend to be more stubborn, this confidence is not necessarily grounded in competence, making it crucial to understand when to trust such social signals.

\section{Discussion}
The primary objective of our work is to understand according to which criteria influence emerges through multi-agent deliberation by casting MASs as mixture of experts (MoE) with adaptive routing.
Our analysis of the router revealed not only confidence but also relative confidence as most dominant factor to explain which agent dominates the opinion formation process.
This suggests that the overall distribution of confidence among agents is highly predictive of the deliberation outcome.
To gain deeper insights into influence emergence, a natural next step in future would be to build more accurate router models considering all, not only the most confident agent.
Graph Neural Networks (GNNs) \citep{wu2020comprehensive,corso2024graph}, in particular, GNNs with attention mechanism \citep{brody2022how,mustafa2023are,mustafa2024gate,mustafa2024dynamic} could be good router candidates, as they are permutation invariant with respect to specific agent identities. 
They also have the potential to generalize across different communication graph topologies 
and even provide an opportunity to design more sophisticated and expressive multi-agent routing in future \citep{daggnn,bause2025maximally,kummer2025weisfeiler,roth2024preventing}.
However, recent insights \citep{rubio-madrigal2026fixed,rubio-madrigal2026linear} imply that more complex communication patterns might not even be required to obtain high performance, especially, if agents are sufficiently complex and if the communication graph can be optimized \citep{jamadandi2024spectral,rubio-madrigal2025gnns}. 
How much can be gained from optimal routing, compared to the routing that emerges from deliberation, remains an open question, that could lead to build more accurate, efficient, and trustworthy agentic systems.

\section{Conclusion}

We have provided a framework to study the mechanisms of collaboration in multi-agent LLM systems (MASs) and gained new insights into the factors governing the emergence of influence, including the confidence of agents, their communication behavior and opinion alignment. 
Building on the observation that the Friedkin-Johnsen model captures MAS deliberation dynamics, we have cast MASs as mixtures of experts (MoE) with adaptive routing.
This has allowed us to derive conditions when MASs can outperform simpler ensembles and single agent models.
Accordingly, their strengths arise from adaptive routing and local specialization of diverse and well calibrated agents,
while its limitations arise from miscalibrated agent confidence, misleading consensus, and routing errors.
We see multiple important directions for future work: 
To understand factors contributing to the emergence of influence, we could learn more advanced routing mechanisms that take relevant features into account.
Graph neural networks  would be a natural approach to identify general routing mechanisms that apply to different agent communication topologies.
Extending the framework to other tasks, such as open-ended generation and tool use, could transfer our insights to different agentic use cases.
This could also aid mitigating overly confident agents and improving communication to enhance overall performance. 

\paragraph{Acknowledgments.}
We are grateful for funding from
the European Research Council (ERC) under the Horizon
Europe Framework Programme (HORIZON) for proposal
number 101116395 SPARSE-ML.

\bibliography{references}
\bibliographystyle{icml2026}

\clearpage
\appendix
\onecolumn

\section{Theory: Proofs of theorems and discussion}

\paragraph{Equilibrium beliefs are convex ensembles.}
The Friedkin-Johnsen dynamics induce a graph-dependent ensemble over the initial agent beliefs.
\begin{proposition}[Graph-induced convex ensemble]
Assume that $\rho(H)<1$, where $\rho(\cdot)$ denotes the spectral radius. Then the FJ dynamics converge to the unique equilibrium $B^\star=(I-H)^{-1}\Gamma S$. $M := (I-H)^{-1}\Gamma$
is nonnegative and row-stochastic. Consequently, each equilibrium belief is a convex combination of the innate beliefs: $   b_i^\star = \sum_{j=1}^n M_{ij}s_j$. 
If the agentic network belief $b_{\mathrm{out}}^\star = \sum_{i=1}^n \eta_i b_i^\star = \sum_{j=1}^n \pi_j s_j$ with $\pi^\top=\eta^\top M$ is formed by a linear aggregation (typically $\eta_i=1/N$) of all agents' beliefs, the agentic network is equivalent to a fixed ensemble with weights $\pi_j$ that depend on the FJ parameters.
\end{proposition}

\begin{proof}[Sketch]
Unrolling the FJ dynamics gives
\[
    B(t)
    =
    \sum_{\tau=0}^{t-1} H^\tau \Gamma S
    +
    H^t B(0).
\]
If $\rho(H)<1$, then $H^t\to 0$ and the Neumann series converges:
\[
    \sum_{\tau=0}^{\infty}H^\tau=(I-H)^{-1}.
\]
Nonnegativity follows because $H$ and $\Gamma$ are nonnegative. Row-stochasticity follows from conservation of mass on the simplex: if all innate beliefs are equal to a constant simplex vector $s$, then all beliefs remain equal to $s$ under the dynamics, implying $M\mathbf 1=\mathbf 1$.
\end{proof}

Thus, deliberation over a fixed graph does not produce an arbitrary new predictor. It produces a graph-induced convex ensemble of the agents' initial beliefs. 
For a designated readout agent $r$, the final prediction is
\begin{equation}
\label{eq:readout-ensemble}
    b_r^\star
    =
    \sum_{j=1}^n m_{rj} s_j.
\end{equation}
More generally, if the system output is a linear readout $\eta\in\Delta^n$ over the final agent beliefs (typically an average with $\eta_i = 1/N$), then
\[
    b_{\mathrm{out}}^\star
    =
    \sum_{i=1}^n \eta_i b_i^\star
    =
    \sum_{j=1}^n \pi_j s_j,
    \qquad
    \pi^\top=\eta^\top M.
\]
Hence, a fixed agentic deliberation graph is equivalent to a fixed convex ensemble over agent priors, with weights determined by the network topology and the agents' stubbornness, memory, and susceptibility parameters.

\paragraph{Static deliberation versus mixture-of-experts deliberation.}

The preceding result shows that a fixed Friedkin--Johnsen network implements a static ensemble:
\begin{equation}
\label{eq:static-ensemble}
    f_{\mathrm{FJ}}(x)
    =
    \sum_{j=1}^n \pi_j s_j(x),
\end{equation}
where $s_j(x)$ is the prediction of agent $j$ on input $x$, and $\pi_j$ is independent of $x$.

If, however, the graph, influence weights, or deliberation parameters depend on the input, then the equilibrium weights also become input-dependent. Writing
\[
    W=W(x),\qquad
    \Gamma=\Gamma(x),\qquad
    A=A(x),
\]
we obtain
\[
    M(x)
    =
    \bigl(I-H(x)\bigr)^{-1}\Gamma(x),
\]
and therefore
\begin{equation}
\label{eq:fj-moe}
    f_{\mathrm{FJ-MoE}}(x)
    =
    \sum_{j=1}^n \pi_j(x)s_j(x).
\end{equation}
This is a mixture-of-experts architecture: the agents are experts, while the graph-dependent deliberation mechanism acts as a router or gating function. 
The static ensemble in \eqref{eq:static-ensemble} can only assign global importance to agents. 
By contrast, the mixture form in \eqref{eq:fj-moe} can assign local, input-dependent responsibility to agents.

\begin{lemma}[Local ambiguity decomposition]
For any belief-dependent weights $a(S)\in\Delta^n$, $    \mathbb E
    \left[
        \left\|
            \sum_{j=1}^n a_j(S)s_j(X)-e_Y
        \right\|_2^2
        \mid S
    \right]
    =
    \sum_{j=1}^n a_j(S) r_j(S)
    -
    D_{a(S)}(S)$.
\end{lemma}
\begin{proof}
For any fixed $S$ and any label $y$,
\[
    \sum_j a_j\|s_j-e_y\|_2^2
    =
    \left\|
        \sum_j a_j s_j-e_y
    \right\|_2^2
    +
    \sum_j a_j
    \left\|
        s_j-\sum_k a_k s_k
    \right\|_2^2.
\]
Taking conditional expectation over $Y\mid S$ gives the claim.
\end{proof}

\begin{proposition}[MoE advantage under mutually exclusive competence]
Assume $n\ge 2$, $\rho_j>0$ for all $j$, and $p>u$. Then $ L_{\mathrm{MoE}}    < \inf_{\psi\in\Delta^N} L_{\mathrm{ens}}(\psi)$.
\end{proposition}
\begin{proof}
For a symmetric prediction that assigns probability $q$ to the correct class and distributes the remaining mass uniformly over the $d-1$ incorrect classes, the Brier loss is
\begin{align}
\begin{split}
\mathcal{B}_d(q)=
(q-1)^2
+
(d-1)
\left(\frac{1-q}{d-1}\right)^2
=
\frac{d}{d-1}(1-q)^2.
\end{split}
\end{align}
Because $q<1$, this function is strictly decreasing in $q$.

On region $\mathcal{X}_j$, a fixed ensemble with weights $\psi$ assigns probability
$
q_j(\psi)=u+(p-u)\psi_j
$

to the correct class. Since $\psi_j\leq 1$ and $p>u$,
$
q_j(\psi)\leq p,
$
with equality if and only if $\psi_j=1$. Therefore,
$
\mathcal{B}_d(q_j(\psi))
\geq
\mathcal{B}_d(p),
$
with equality if and only if $\psi_j=1$.

For the fixed ensemble to match the MoE on every positive-probability region, it would need
$\psi_j=1$ for every $j$. This is impossible when $n\geq 2$ and
$\sum_j\psi_j=1$. Consequently, at least one region with $\rho_j>0$ satisfies
$q_j(\psi)<p$, and hence
\[
\begin{split}
L_{\mathrm{ens}}(\psi)
=
\sum_{j=1}^n
\rho_j\mathcal{B}_d(q_j(\psi))
 >
\sum_{j=1}^n
\rho_j\mathcal{B}_d(p)
=
\mathcal{B}_d(p)
=
L_{\mathrm{MoE}}.
\end{split}
\]
This holds for every $\psi\in\Delta_n$.
\end{proof}

\begin{proof}[Optimal ensemble for balanced mutually exclusive regions.]
When $\rho_j=1/n$, define
\[
h(t) =\mathcal{B}_d\left(u+(p-u)t\right).
\]
Since $\mathcal{B}_d$ is strictly convex and $p\neq u$, $h$ is strictly convex. Jensen's inequality
gives
\[
\begin{split}
L_{\mathrm{ens}}(\psi) =
\frac{1}{n}\sum_{j=1}^n h(\psi_j)
\geq
h\left(\frac{1}{n}\sum_{j=1}^n\psi_j\right)
=
h\left(\frac{1}{n}\right).
\end{split}
\]
Equality holds if and only if
\[
\psi_1=\cdots=\psi_n=\frac{1}{n}.
\]
Thus, the uniform ensemble is the unique optimum and
\[
L_{\mathrm{ens}}^\star
=
\mathcal{B}_d\left(
u+\frac{p-u}{n}
\right).
\]
Subtracting $L_{\mathrm{MoE}}=\mathcal{B}_d(p)$ gives the stated gap.
\end{proof}

\begin{proof}[Comparison with the best single agent]
Agent $j$ obtains loss $\mathcal{B}_d(p)$ on its competence region and
$\mathcal{B}_d(u)$ elsewhere. Thus,
\[
L_{\mathrm{single},j}
=
\rho_j\mathcal{B}_d(p)
+
(1-\rho_j)\mathcal{B}_d(u).
\]
Because $\mathcal{B}_d(p)<\mathcal{B}_d(u)$, this expression is minimized by an agent with maximum
$\rho_j$. Therefore,
\[
\begin{split}
L{\mathrm{single}}^\star-L{\mathrm{MoE}}
& =
\rho{j^{\mathrm{glob}}}\mathcal{B}_d(p)
+
\bigl(1-\rho{j^{\mathrm{glob}}}\bigr)\mathcal{B}_d(u)
-
\mathcal{B}_d(p)
\\
& =
\bigl(1-\rho{j^{\mathrm{glob}}}\bigr)
\left[
\mathcal{B}_d(u)-\mathcal{B}_d(p)
\right].
\end{split}
\]
This is strictly positive whenever no single region has probability one.
\end{proof}

\begin{proof}[Routing-error threshold]
A router that selects the competent agent with probability $1-\delta$ and an uninformative agent
with probability $\delta$ has expected loss
\[
L_{\mathrm{router}}(\delta)
=
(1-\delta)\mathcal{B}_d(p)
+
\delta\mathcal{B}_d(u).
\]
In the balanced case, the optimal fixed ensemble has loss
\[
L_{\mathrm{ens}}^\star
=
\mathcal{B}_d(q_n),
\qquad
q_n=u+\frac{p-u}{n}.
\]
The router is better if and only if
\[
\mathcal{B}_d(p)
+
\delta
\left[
\mathcal{B}_d(u)-\mathcal{B}_d(p)
\right]
<
\mathcal{B}_d(q_n).
\]
Since $\mathcal{B}_d(u)>\mathcal{B}_d(p)$, division by the positive denominator gives
\[
\delta
<
\frac{
\mathcal{B}_d(q_n)-\mathcal{B}_d(p)
}{
\mathcal{B}_d(u)-\mathcal{B}_d(p)
}.
\]
Substituting
\[
\mathcal{B}_d(q)=\frac{d}{d-1}(1-q)^2
\]
and cancelling the common factor gives
\[
\delta
<
\frac{
(1-q_n)^2-(1-p)^2
}{
(1-u)^2-(1-p)^2
}.
\]
\end{proof}

\begin{proposition}
Assume $n\geq 2$, $d\geq 3$, and the conditions above. Then the optimal fixed ensemble is the uniform ensemble with $\psi_k^\star=\frac{1}{n}$ and confidence routing strictly improves over it:
$L_{\mathrm{MoE}} < L_{\mathrm{ens}}(\psi^\star)$. 
More precisely, with $t=\frac{n-1}{n}$, the loss gap is
\begin{align*}
L_{\mathrm{ens}}(\psi^\star)-L_{\mathrm{MoE}}
=
2t(p-u)(1-p+b) + t^2
\left[
(p-u)^2
+(c-b)^2
+(d-2)(r-b)^2
\right]
>0. 
\end{align*} 
\end{proposition}
\begin{proof}
On region $\mathcal{X}_k$, let $v_k$ denote the competent agent's prediction and $w_k$ the common
prediction of every non-competent agent. For any fixed ensemble $\psi$, its prediction on this
region is
\[
q_k(\psi_k)
=
\psi_k v_k+(1-\psi_k)w_k.
\]
Define
\[
h(t)
=
\left\lVert
tv_k+(1-t)w_k-e_{y_k}
\right\rVert_2^2.
\]
The function $h$ is strictly convex because it is a squared norm of an affine function and
$v_k\neq w_k$. By symmetry, the same function applies to every region. Hence,
\[
L_{\mathrm{ens}}(\psi)
=
\frac{1}{n}\sum_{k=1}^n h(\psi_k).
\]
Jensen's inequality gives
\[
L_{\mathrm{ens}}(\psi)
\geq
h\left(
\frac{1}{n}\sum_{k=1}^n\psi_k
\right)
=
h\left(\frac{1}{n}\right),
\]
with equality if and only if $\psi_k=1/n$ for every $k$. Thus, the uniform ensemble is the unique
optimal fixed ensemble.

Let
\[
t=\frac{n-1}{n}.
\]
The uniform ensemble prediction on region $\mathcal{X}_k$ is
\[
\bar q_k
=
\frac{1}{n}v_k+\frac{n-1}{n}w_k
=
v_k+t(w_k-v_k).
\]
Its excess Brier loss over the competent agent is
\[
\begin{split}
\lVert\bar q_k-e{y_k}\rVert_2^2
-
\lVert v_k-e_{y_k}\rVert_2^2
& =
2t
\left\langle
v_k-e_{y_k},
w_k-v_k
\right\rangle
\\
&\quad
+
t^2\lVert w_k-v_k\rVert_2^2.
\end{split}
\]

The inner product is
\[
\begin{split}
\left\langle
v_k-e_{y_k},
w_k-v_k
\right\rangle
&=
(p-1)(u-p)
+
b\left[(1-u)-(1-p)\right]
\\
&=
(1-p)(p-u)+b(p-u)
\\
&=
(p-u)(1-p+b).
\end{split}
\]
Furthermore,
\[
\lVert w_k-v_k\rVert_2^2
=
(p-u)^2
+(c-b)^2
+(d-2)(r-b)^2.
\]
Therefore,
\[
\begin{split}
L_{\mathrm{ens}}(\psi^\star)-L_{\mathrm{MoE}}
&=
2t(p-u)(1-p+b)
\\
&\quad
+
t^2
\left[
(p-u)^2
+(c-b)^2
+(d-2)(r-b)^2
\right].
\end{split}
\]
Since $t>0$ and $p>u$, the first term is strictly positive, while the second is nonnegative. Hence,
\[
L_{\mathrm{ens}}(\psi^\star)
>
L_{\mathrm{MoE}}.
\]

Finally, under the uniform ensemble,
\[
\bar q_y
=
\frac{p+(n-1)u}{n},
\qquad
\bar q_z
=
\frac{b+(n-1)c}{n}.
\]
Because $c>r$, the common wrong label receives more probability than every other wrong label.
Thus, the ensemble predicts the wrong label whenever
\[
b+(n-1)c
>
p+(n-1)u.
\]
The confidence-routed MoE predicts correctly because
\[
p>\frac{1}{d}
\quad\Longleftrightarrow\quad
p>\frac{1-p}{d-1}=b.
\]
\end{proof}

\section{Experimental evaluation}
\label{app:accuracy_res}
We evaluate how well the FJ model fits in detail and present the results in Table~\ref{tab:model_fit_app}. Over all datasets, we see a good fit.

\begin{table}[h]
\centering
\caption{Model fit (KL divergence and MSE) across datasets, prompt types, and models.}
\label{tab:model_fit_app}
\begin{tabular}{lllcc}
\toprule
Dataset & Prompt & Model & KL & MSE \\
\midrule
\multirow{12}{*}{MMLU-Pro} & \multirow{3}{*}{neutral} & \textit{GPT-5.4 Mini} & $0.0829 \pm 0.0147$ & $0.00247 \pm 0.00097$ \\
 &  & \textit{Qwen2.5-14B-Instruct} & $0.0503 \pm 0.0013$ & $0.00201 \pm 0.00004$ \\
 &  & \textit{Qwen2.5-72B-Instruct} & $0.0282 \pm 0.0040$ & $0.000972 \pm 0.000176$ \\[3pt]
 & \multirow{3}{*}{experts} & \textit{GPT-5.4 Mini} & $0.0714 \pm 0.0104$ & $0.00206 \pm 0.00040$ \\
 &  & \textit{Qwen2.5-14B-Instruct} & $0.0467 \pm 0.0044$ & $0.00182 \pm 0.00032$ \\
 &  & \textit{Qwen2.5-72B-Instruct} & $0.0268 \pm 0.0068$ & $0.000935 \pm 0.000226$ \\[3pt]
 & \multirow{3}{*}{comm. styles} & \textit{GPT-5.4 Mini} & $0.0740 \pm 0.0096$ & $0.00219 \pm 0.00052$ \\
 &  & \textit{Qwen2.5-14B-Instruct} & $0.0505 \pm 0.0148$ & $0.00200 \pm 0.00067$ \\
 &  & \textit{Qwen2.5-72B-Instruct} & $0.0263 \pm 0.0077$ & $0.000859 \pm 0.000326$ \\[3pt]
 & \multirow{3}{*}{roles} & \textit{GPT-5.4 Mini} & $0.0786 \pm 0.0212$ & $0.00227 \pm 0.00090$ \\
 &  & \textit{Qwen2.5-14B-Instruct} & $0.0575 \pm 0.0061$ & $0.00225 \pm 0.00032$ \\
 &  & \textit{Qwen2.5-72B-Instruct} & $0.0320 \pm 0.0117$ & $0.00114 \pm 0.00049$ \\
\midrule
\multirow{12}{*}{BBQ} & \multirow{3}{*}{neutral} & \textit{GPT-5.4 Mini} & $0.0420 \pm 0.0093$ & $0.000373 \pm 0.000100$ \\
 &  & \textit{Qwen2.5-14B-Instruct} & $0.0403 \pm 0.0043$ & $0.00117 \pm 0.00006$ \\
 &  & \textit{Qwen2.5-72B-Instruct} & $0.0238 \pm 0.0005$ & $0.000876 \pm 0.000135$ \\[3pt]
 & \multirow{3}{*}{experts} & \textit{GPT-5.4 Mini} & $0.0423 \pm 0.0075$ & $0.000324 \pm 0.000091$ \\
 &  & \textit{Qwen2.5-14B-Instruct} & $0.0541 \pm 0.0076$ & $0.00151 \pm 0.00053$ \\
 &  & \textit{Qwen2.5-72B-Instruct} & $0.0239 \pm 0.0152$ & $0.000963 \pm 0.000827$ \\[3pt]
 & \multirow{3}{*}{comm. styles} & \textit{GPT-5.4 Mini} & $0.0459 \pm 0.0052$ & $0.000436 \pm 0.000063$ \\
 &  & \textit{Qwen2.5-14B-Instruct} & $0.0354 \pm 0.0039$ & $0.00105 \pm 0.00011$ \\
 &  & \textit{Qwen2.5-72B-Instruct} & $0.0221 \pm 0.0016$ & $0.000837 \pm 0.000125$ \\[3pt]
 & \multirow{3}{*}{roles} & \textit{GPT-5.4 Mini} & $0.0479 \pm 0.0039$ & $0.000422 \pm 0.000135$ \\
 &  & \textit{Qwen2.5-14B-Instruct} & $0.0508 \pm 0.0017$ & $0.00169 \pm 0.00031$ \\
 &  & \textit{Qwen2.5-72B-Instruct} & $0.0214 \pm 0.0057$ & $0.000817 \pm 0.000305$ \\
\midrule
\multirow{12}{*}{CSQA} & \multirow{3}{*}{neutral} & \textit{GPT-5.4 Mini} & $0.0513 \pm 0.0040$ & $0.00247 \pm 0.00019$ \\
 &  & \textit{Qwen2.5-14B-Instruct} & $0.0754 \pm 0.0146$ & $0.00570 \pm 0.00197$ \\
 &  & \textit{Qwen2.5-72B-Instruct} & $0.0364 \pm 0.0034$ & $0.00264 \pm 0.00021$ \\[3pt]
 & \multirow{3}{*}{experts} & \textit{GPT-5.4 Mini} & $0.0517 \pm 0.0015$ & $0.00256 \pm 0.00027$ \\
 &  & \textit{Qwen2.5-14B-Instruct} & $0.0599 \pm 0.0106$ & $0.00414 \pm 0.00110$ \\
 &  & \textit{Qwen2.5-72B-Instruct} & $0.0331 \pm 0.0019$ & $0.00238 \pm 0.00004$ \\[3pt]
 & \multirow{3}{*}{comm. styles} & \textit{GPT-5.4 Mini} & $0.0516 \pm 0.0042$ & $0.00244 \pm 0.00020$ \\
 &  & \textit{Qwen2.5-14B-Instruct} & $0.0647 \pm 0.0086$ & $0.00443 \pm 0.00047$ \\
 &  & \textit{Qwen2.5-72B-Instruct} & $0.0316 \pm 0.0026$ & $0.00222 \pm 0.00026$ \\[3pt]
 & \multirow{3}{*}{roles} & \textit{GPT-5.4 Mini} & $0.0534 \pm 0.0037$ & $0.00272 \pm 0.00032$ \\
 &  & \textit{Qwen2.5-14B-Instruct} & $0.0770 \pm 0.0121$ & $0.00592 \pm 0.00133$ \\
 &  & \textit{Qwen2.5-72B-Instruct} & $0.0306 \pm 0.0049$ & $0.00223 \pm 0.00063$ \\
\bottomrule
\end{tabular}

\end{table}

We compare the performance of the agentic system to the baseline of taking the answer with maximum belief (averaging over the initial agent beliefs), and using a task independent ensemble by fitting constant FJ parameters over all samples, optimizing for predictive performance. Table~\ref{tab:global_ensemble_accuracy} that agentic systems can perform better than both baseline and FJ ensemble, underlining the theory that the routing is input dependent, suggesting a MoE model. 
In some cases, the agentic system slightly underperforms the baseline when prompted with specific persona or communication style prompts. Similar behavior was described in~\cite{hu2026expertpersonasimprovellm}, which shows personas can potentially damage the discriminative performance. Regardless, utilizing these prompts introduces diversity in the agentic system, enabling a better study of its dynamics.

In rare cases $(<1\%)$, the FJ solver fails to converge within a reasonable time, so we skip the affected sample. In $4\%$ of the successful samples, the fitted FJ parameters result in the condition number of $(I-H) \geq 10^4$, consequenlty $\pi_j$'s are noisy and no longer sum exactly to one, so we exclude these cases from the influence regressions and plots. 

\begin{table}[tb]
\centering
\caption{Accuracy of baseline, FJ ensemble, and agent system, across datasets, prompts, and models.}
\label{tab:global_ensemble_accuracy}
\footnotesize
\begin{tabular}{lllccc}
\toprule
Dataset & Prompt & Model & Baseline & FJ Ens. & MAS \\
\midrule
\multirow{12}{*}{MMLU-Pro} & \multirow{3}{*}{neutral} & \textit{GPT-5.4 Mini} & $0.804\pm0.013$ & $0.820\pm0.013$ & $\mathbf{0.826\pm0.027}$ \\
 &  & \textit{Qwen2.5-14B-Instruct} & $0.569\pm0.027$ & $0.577\pm0.034$ & $\mathbf{0.604\pm0.023}$ \\
 &  & \textit{Qwen2.5-72B-Instruct} & $0.695\pm0.036$ & $0.694\pm0.075$ & $\mathbf{0.712\pm0.016}$ \\[3pt]
 & \multirow{3}{*}{experts} & \textit{GPT-5.4 Mini} & $0.816\pm0.005$ & $0.819\pm0.009$ & $\mathbf{0.835\pm0.008}$ \\
 &  & \textit{Qwen2.5-14B-Instruct} & $0.637\pm0.019$ & $0.633\pm0.034$ & $\mathbf{0.640\pm0.016}$ \\
 &  & \textit{Qwen2.5-72B-Instruct} & $0.728\pm0.043$ & $0.729\pm0.007$ & $\mathbf{0.742\pm0.037}$ \\[3pt]
 & \multirow{3}{*}{comm. styles} & \textit{GPT-5.4 Mini} & $0.797\pm0.012$ & $0.809\pm0.029$ & $\mathbf{0.824\pm0.029}$ \\
 &  & \textit{Qwen2.5-14B-Instruct} & $0.591\pm0.023$ & $\mathbf{0.614\pm0.029}$ & $0.606\pm0.016$ \\
 &  & \textit{Qwen2.5-72B-Instruct} & $0.727\pm0.019$ & $0.736\pm0.017$ & $\mathbf{0.743\pm0.015}$ \\[3pt]
 & \multirow{3}{*}{roles} & \textit{GPT-5.4 Mini} & $0.805\pm0.023$ & $0.806\pm0.029$ & $\mathbf{0.824\pm0.004}$ \\
 &  & \textit{Qwen2.5-14B-Instruct} & $0.557\pm0.014$ & $0.576\pm0.018$ & $\mathbf{0.596\pm0.009}$ \\
 &  & \textit{Qwen2.5-72B-Instruct} & $0.700\pm0.040$ & $0.706\pm0.045$ & $\mathbf{0.720\pm0.060}$ \\
\midrule
\multirow{12}{*}{BBQ} & \multirow{3}{*}{neutral} & \textit{GPT-5.4 Mini} & $0.916\pm0.014$ & $\mathbf{0.916\pm0.028}$ & $0.916\pm0.024$ \\
 &  & \textit{Qwen2.5-14B-Instruct} & $0.866\pm0.028$ & $0.867\pm0.015$ & $\mathbf{0.876\pm0.021}$ \\
 &  & \textit{Qwen2.5-72B-Instruct} & $0.882\pm0.002$ & $\mathbf{0.888\pm0.016}$ & $0.875\pm0.015$ \\[3pt]
 & \multirow{3}{*}{experts} & \textit{GPT-5.4 Mini} & $0.932\pm0.017$ & $0.932\pm0.012$ & $\mathbf{0.934\pm0.012}$ \\
 &  & \textit{Qwen2.5-14B-Instruct} & $0.880\pm0.017$ & $\mathbf{0.886\pm0.009}$ & $0.870\pm0.012$ \\
 &  & \textit{Qwen2.5-72B-Instruct} & $0.836\pm0.015$ & $\mathbf{0.858\pm0.031}$ & $0.840\pm0.006$ \\[3pt]
 & \multirow{3}{*}{comm. styles} & \textit{GPT-5.4 Mini} & $0.906\pm0.025$ & $0.903\pm0.015$ & $\mathbf{0.916\pm0.014}$ \\
 &  & \textit{Qwen2.5-14B-Instruct} & $0.865\pm0.005$ & $0.882\pm0.020$ & $\mathbf{0.887\pm0.018}$ \\
 &  & \textit{Qwen2.5-72B-Instruct} & $0.857\pm0.020$ & $\mathbf{0.875\pm0.034}$ & $0.870\pm0.025$ \\[3pt]
 & \multirow{3}{*}{roles} & \textit{GPT-5.4 Mini} & $0.904\pm0.018$ & $\mathbf{0.918\pm0.018}$ & $0.915\pm0.018$ \\
 &  & \textit{Qwen2.5-14B-Instruct} & $0.874\pm0.034$ & $0.866\pm0.026$ & $\mathbf{0.898\pm0.028}$ \\
 &  & \textit{Qwen2.5-72B-Instruct} & $0.860\pm0.009$ & $\mathbf{0.887\pm0.013}$ & $0.882\pm0.018$ \\
\midrule
\multirow{12}{*}{CSQA} & \multirow{3}{*}{neutral} & \textit{GPT-5.4 Mini} & $\mathbf{0.726\pm0.013}$ & $0.722\pm0.010$ & $0.722\pm0.017$ \\
 &  & \textit{Qwen2.5-14B-Instruct} & $0.775\pm0.011$ & $\mathbf{0.782\pm0.013}$ & $0.782\pm0.018$ \\
 &  & \textit{Qwen2.5-72B-Instruct} & $\mathbf{0.803\pm0.012}$ & $0.799\pm0.028$ & $0.793\pm0.011$ \\[3pt]
 & \multirow{3}{*}{experts} & \textit{GPT-5.4 Mini} & $0.723\pm0.014$ & $\mathbf{0.733\pm0.014}$ & $0.730\pm0.025$ \\
 &  & \textit{Qwen2.5-14B-Instruct} & $0.787\pm0.014$ & $\mathbf{0.803\pm0.014}$ & $0.797\pm0.029$ \\
 &  & \textit{Qwen2.5-72B-Instruct} & $0.806\pm0.017$ & $\mathbf{0.816\pm0.013}$ & $0.803\pm0.016$ \\[3pt]
 & \multirow{3}{*}{comm. styles} & \textit{GPT-5.4 Mini} & $0.730\pm0.025$ & $\mathbf{0.733\pm0.052}$ & $0.730\pm0.025$ \\
 &  & \textit{Qwen2.5-14B-Instruct} & $0.777\pm0.014$ & $\mathbf{0.787\pm0.038}$ & $0.763\pm0.057$ \\
 &  & \textit{Qwen2.5-72B-Instruct} & $0.799\pm0.047$ & $0.806\pm0.013$ & $\mathbf{0.806\pm0.036}$ \\[3pt]
 & \multirow{3}{*}{roles} & \textit{GPT-5.4 Mini} & $0.727\pm0.029$ & $\mathbf{0.737\pm0.014}$ & $0.720\pm0.000$ \\
 &  & \textit{Qwen2.5-14B-Instruct} & $0.777\pm0.014$ & $0.793\pm0.052$ & $\mathbf{0.803\pm0.014}$ \\
 &  & \textit{Qwen2.5-72B-Instruct} & $\mathbf{0.809\pm0.027}$ & $0.806\pm0.030$ & $0.799\pm0.003$ \\
\bottomrule
\end{tabular}
\end{table}

\section{Additional Plots}
\label{app:addplots}
We present additional plots for our experiments. More results can be found in the supplementary material. 
We see the same variability of the FJ parameters and tendency towards consensus in the other datasets and when using \emph{Qwen2.5-14B-Instruct} and \textit{Qwen2.5-72B-Instruct-GPTQ-Int8} as the underlying models (Figs.~\ref{fig:app:para_var_gpt_neutral},~\ref{fig:app:para_var_qwen_neutral} and~\ref{fig:app:para_var_qwenl_neutral}).

Figs.~\ref{fig:app:con_inf_gpt_neutral},~\ref{fig:app:con_inf_qwen_neutral}, and~\ref{fig:app:con_inf_qwenl_neutral} show the (peer-)influence of agents in relation to their confidence. We see that confident agents tend to be more influential.

Fig.~\ref{fig:app:change_rate_neutral} shows the probability of an agent changing their answer to the majority answer when the agent is confident/not confident. We see that agents are more likely to change their answer if they are not confident. Interestingly, between the models, the \emph{Qwen2.5-14B-Instruct} agents are less confident, but also less likely to change their answer to the majority answer when initially disagreeing with it.

Figs.~\ref{fig:app:heatmaps_gpt},~\ref{fig:app:heatmaps_qwen}, and~\ref{fig:app:heatmaps_qwenl} show the impact of different diversifying prompts on the peer influence of the agents, grouped by question (super-)category on dataset MMLU-Pro. The influence of an agent seems to be driven also by perceived confidence, and roles seem to have a relatively big effect that is consistent over different models.

Fig.~\ref{fig:regression_class} shows coefficient for random forests regression influence and classifying the most influential agent in an MAS. We observe a good fit for regression and a near perfect fit on the classification task. This shows, the defined variables are predictive of FJ dynamics.

We investigate the relation between confidence and competence, as well as between influence and stubbornness in Fig.~\ref{fig:comp_conf_1}. We can see that an agents influence is highly correlated with its stubbornness. The relation between actual competence and influence is less prominent (see~Fig.~\ref{fig:infl_comp}).

Figs.~\ref{fig:app:confusion_matrices_gpt} and~\ref{fig:app:confusion_matrices_qwenl} show the number of cases where Eq.~\ref{eq:MASvsSingle} of Theorem~\ref{th:mas_vs_single} is satisfied for individual samples, compared to when it does not, and whether the MAS outperforms the best single agent. Theorem~\ref{th:mas_vs_single} states that, if Eq.\ref{eq:MASvsSingle} is satisfied, the MAS should outperform the best single agent, so large values in the upper left and lower right entries of the confusion matrix point to Theorem~\ref{th:mas_vs_single} holding.
We can see, that in most cases, the equation is satisfied, and routing wins for \textit{GPT-5.4 Mini}. For \emph{Qwen2.5-72B-Instruct-GPTQ-Int8} the MAS loses more often, but in most of these cases Eq.~\ref{eq:MASvsSingle} is not satisfied.

Figs.~\ref{fig:samples_heatmap} and~\ref{fig:samples_heatmap_2} show weight heatmaps for different samples of the MMLU-Pro dataset (with communication style prompts). We can see that the weight matrices are task dependent.

\begin{figure}
    \centering
    \begin{subfigure}{0.7\textwidth}
        \includegraphics[width=\textwidth]{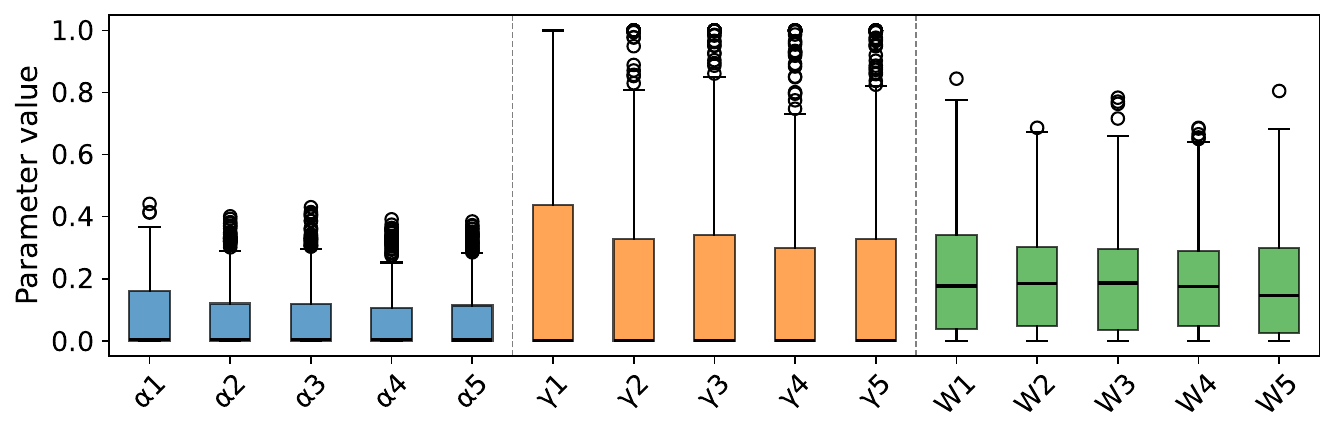}
        \subcaption{MMLU-Pro}
        \label{fig:mmlu_pro_gpt_n_var}
    \end{subfigure}
    \begin{subfigure}{0.25\textwidth}
        \includegraphics[width=\textwidth]{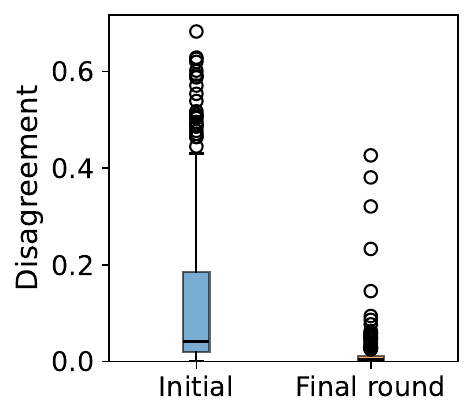}
        \subcaption{Belief disagreement.}
        \label{fig:mmlu_pro_concensus_box_gpt_n}
    \end{subfigure}
    \begin{subfigure}{0.7\textwidth}
        \includegraphics[width=\textwidth]{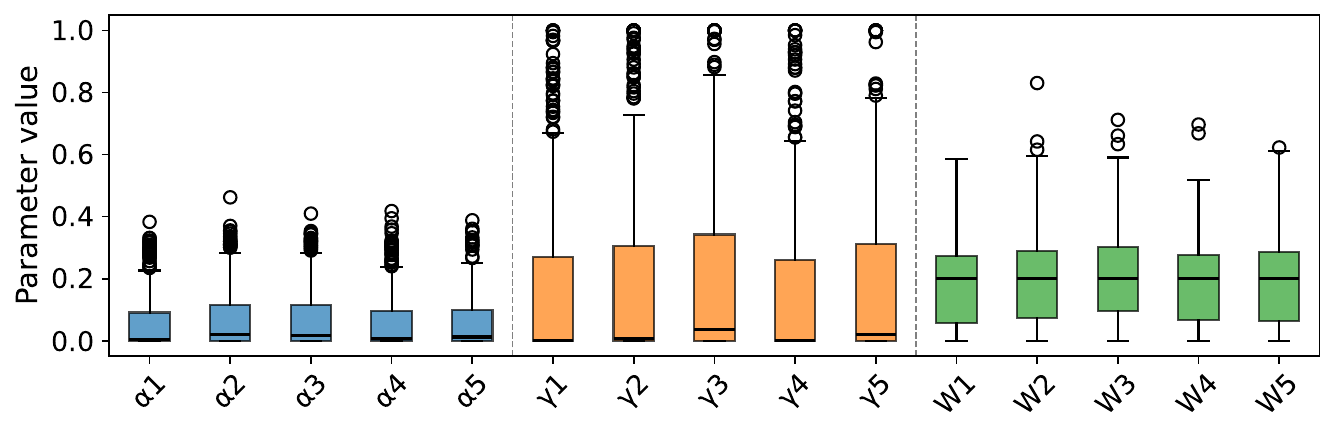}
        \subcaption{BBQ}
        \label{fig:bbq_gpt_n_var}
    \end{subfigure}
    \begin{subfigure}{0.25\textwidth}
        \includegraphics[width=\textwidth]{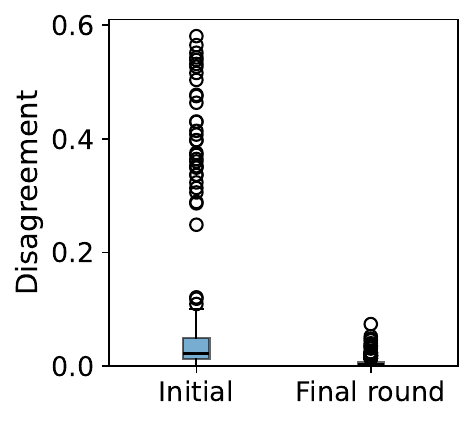}
        \subcaption{Belief disagreement.}
        \label{fig:bbq_concensus_box_gpt_n}
    \end{subfigure}
    \begin{subfigure}{0.7\textwidth}
        \includegraphics[width=\textwidth]{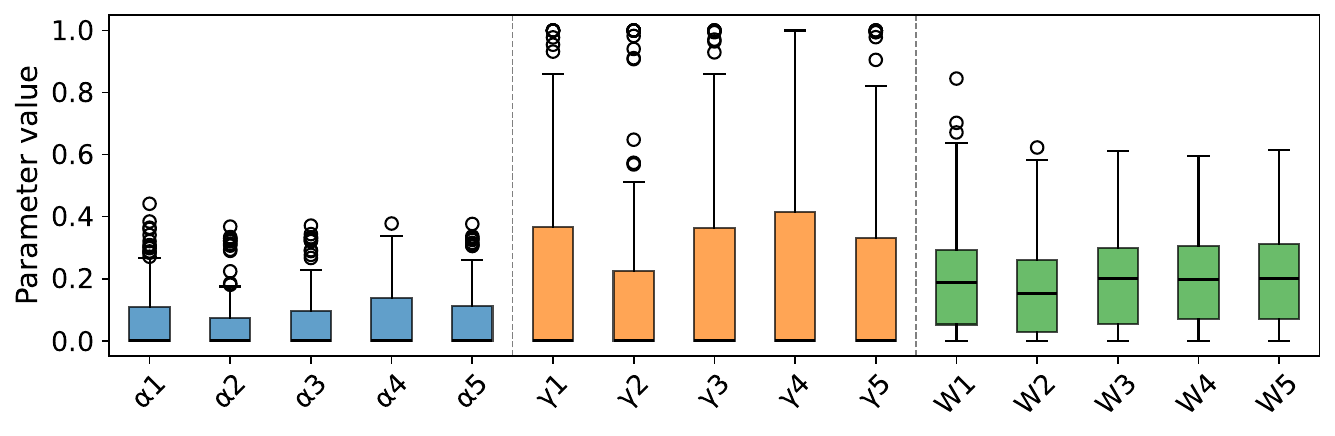}
        \subcaption{CSQA}
        \label{fig:csqa_gpt_n_var}
    \end{subfigure}
    \begin{subfigure}{0.25\textwidth}
        \includegraphics[width=\textwidth]{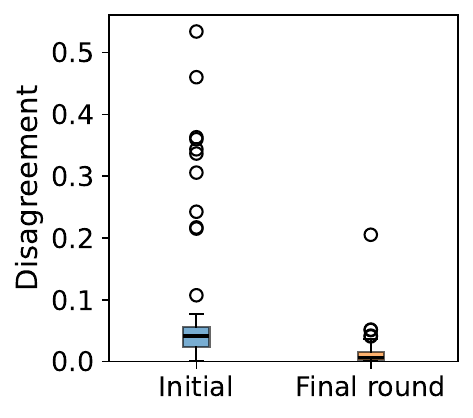}
        \subcaption{Belief disagreement.}
        \label{fig:csqa_concensus_box_gpt_n}
    \end{subfigure}

\caption{Parameter variability over all samples no special prompts for \textit{GPT-5.4 Mini}, receiving weights (W1-5) are averaged over senders (left). Tendency towards consensus for the corresponding dataset (right).}
    \label{fig:app:para_var_gpt_neutral}
\end{figure}

\begin{figure}
    \centering
    \begin{subfigure}{0.7\textwidth}
        \includegraphics[width=\textwidth]{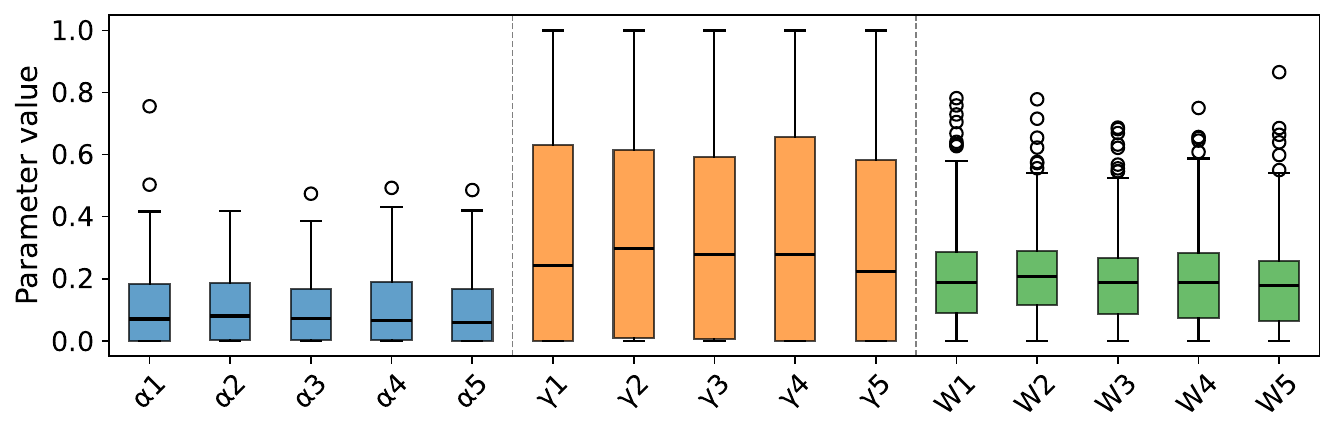}
        \subcaption{MMLU-Pro}
        \label{fig:mmlu_pro_qwen_n_var}
    \end{subfigure}
    \begin{subfigure}{0.25\textwidth}
        \includegraphics[width=\textwidth]{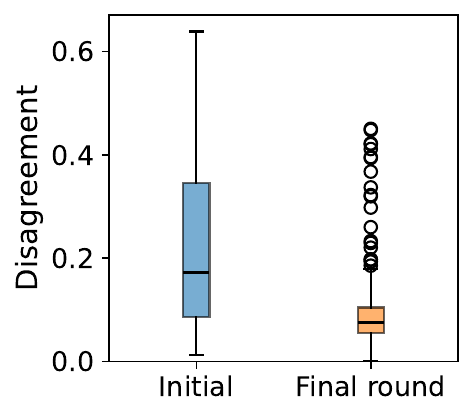}
        \subcaption{Belief disagreement.}
        \label{fig:mmlu_pro_concensus_box_qwen_n}
    \end{subfigure}
    \begin{subfigure}{0.7\textwidth}
        \includegraphics[width=\textwidth]{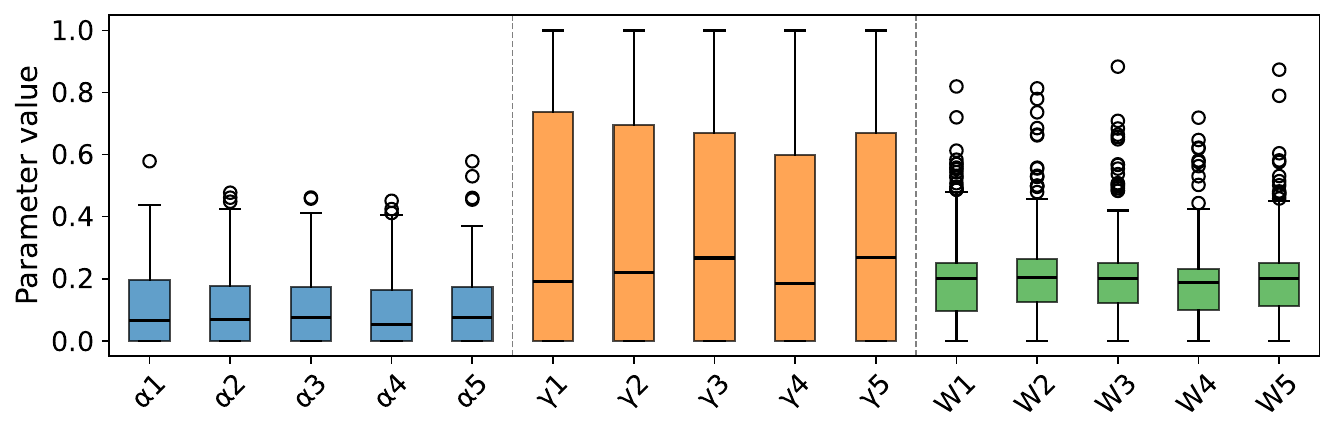}
        \subcaption{BBQ}
        \label{fig:bbq_qwen_n_var}
    \end{subfigure}
    \begin{subfigure}{0.25\textwidth}
        \includegraphics[width=\textwidth]{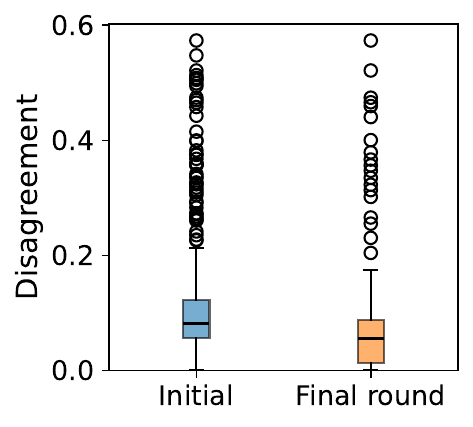}
        \subcaption{Belief disagreement.}
        \label{fig:bbq_concensus_box_qwen_n}
    \end{subfigure}
    \begin{subfigure}{0.7\textwidth}
        \includegraphics[width=\textwidth]{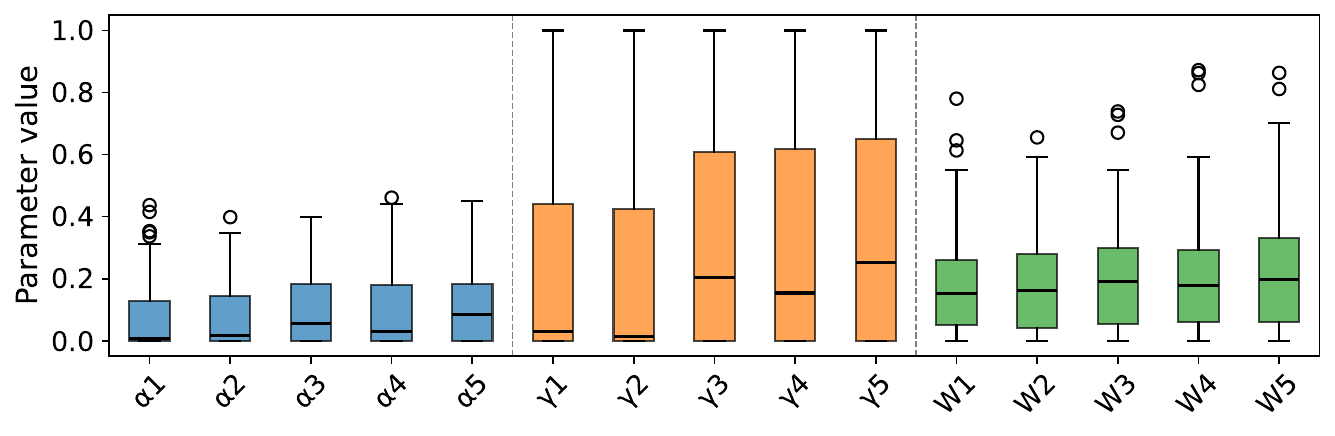}
        \subcaption{CSQA}
        \label{fig:csqa_qwen_n_var}
    \end{subfigure}
    \begin{subfigure}{0.25\textwidth}
        \includegraphics[width=\textwidth]{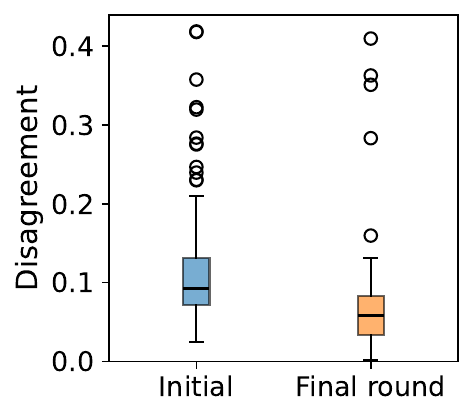}
        \subcaption{Belief disagreement.}
        \label{fig:csqa_concensus_box_qwen_n}
    \end{subfigure}

\caption{Parameter variability over all samples no special prompts for \emph{Qwen2.5-14B-Instruct}, receiving weights (W1-5) are averaged over senders (left). Tendency towards consensus for the corresponding dataset (right).}
    \label{fig:app:para_var_qwen_neutral}
\end{figure}

\begin{figure}
    \centering
    \begin{subfigure}{0.7\textwidth}
        \includegraphics[width=\textwidth]{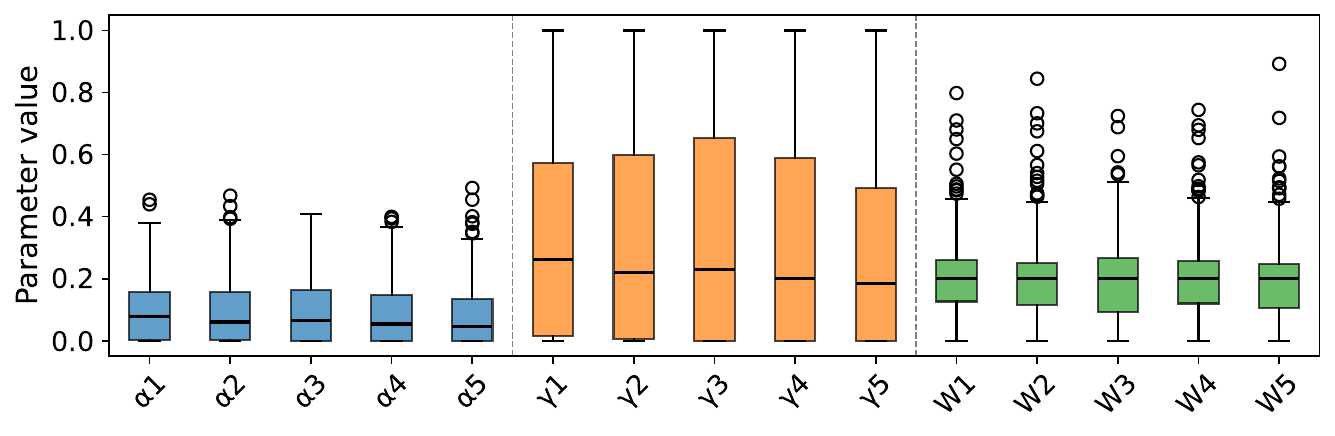}
        \subcaption{MMLU-Pro}
        \label{fig:mmlu_pro_qwenl_n_var}
    \end{subfigure}
    \begin{subfigure}{0.25\textwidth}
        \includegraphics[width=\textwidth]{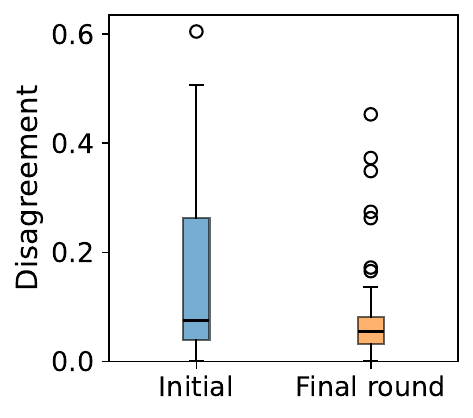}
        \subcaption{Belief disagreement.}
        \label{fig:mmlu_pro_concensus_box_qwenl_n}
    \end{subfigure}
    \begin{subfigure}{0.7\textwidth}
        \includegraphics[width=\textwidth]{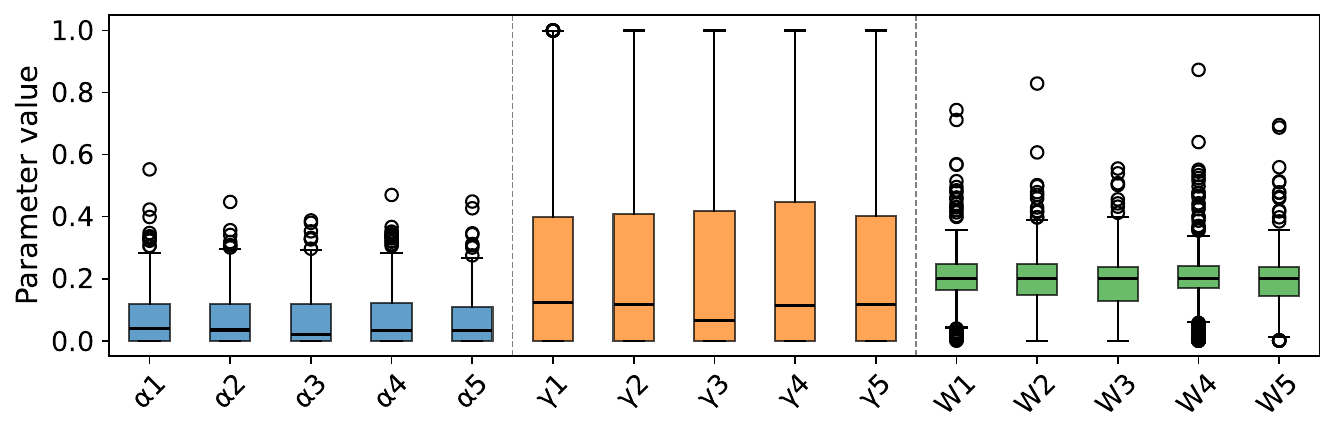}
        \subcaption{BBQ}
        \label{fig:bbq_qwenl_n_var}
    \end{subfigure}
    \begin{subfigure}{0.25\textwidth}
        \includegraphics[width=\textwidth]{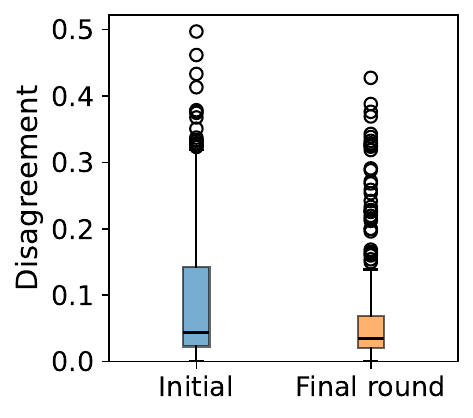}
        \subcaption{Belief disagreement.}
        \label{fig:bbq_concensus_box_qwenl_n}
    \end{subfigure}
    \begin{subfigure}{0.7\textwidth}
        \includegraphics[width=\textwidth]{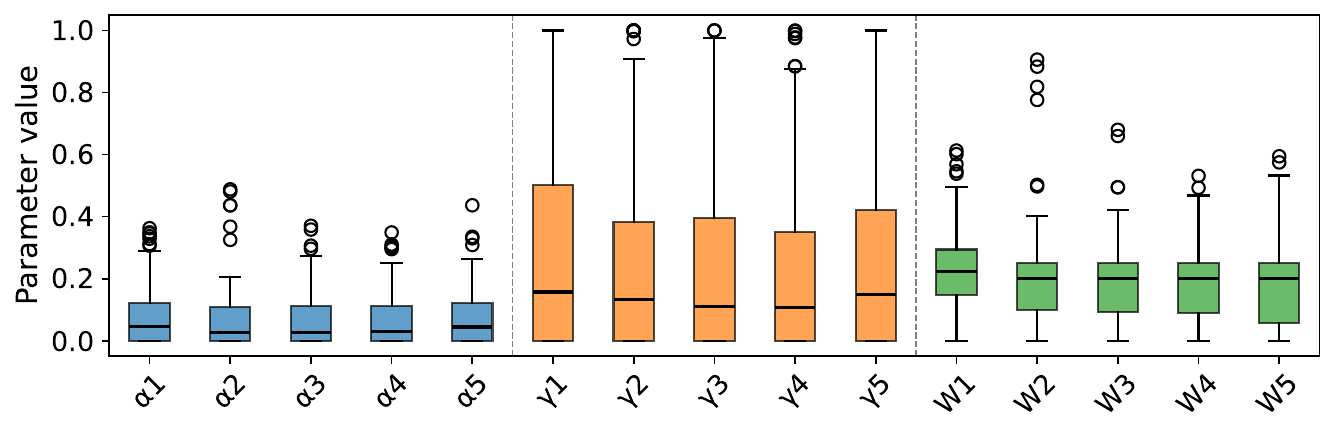}
        \subcaption{CSQA}
        \label{fig:csqa_qwenl_n_var}
    \end{subfigure}
    \begin{subfigure}{0.25\textwidth}
        \includegraphics[width=\textwidth]{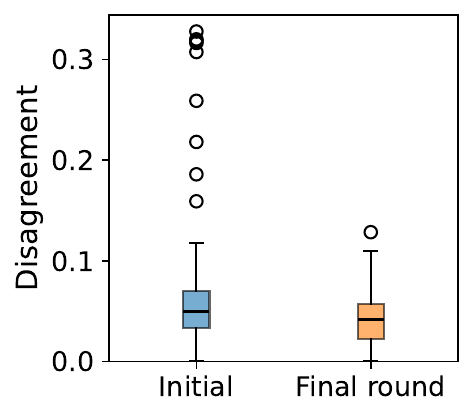}
        \subcaption{Belief disagreement.}
        \label{fig:csqa_concensus_box_qwenl_n}
    \end{subfigure}

\caption{Parameter variability over all samples no special prompts for \emph{Qwen2.5-72B-Instruct-GPTQ-Int8}, receiving weights (W1-5) are averaged over senders (left). Tendency towards consensus for the corresponding dataset (right).}
    \label{fig:app:para_var_qwenl_neutral}
\end{figure}

\begin{figure}
    \centering
    \begin{subfigure}{\textwidth}
    \centering
    \includegraphics[width=0.45\linewidth]{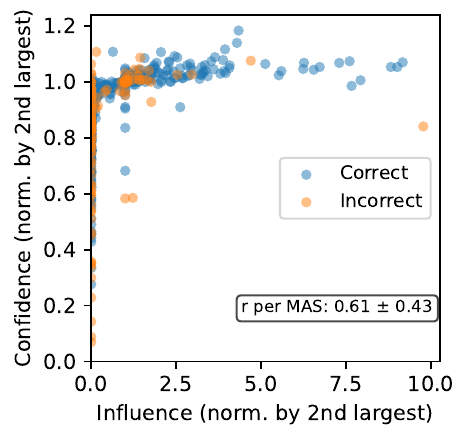}
    \hfill
    \includegraphics[width=0.45\linewidth]{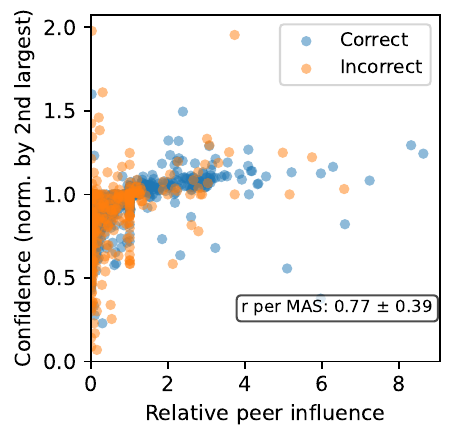}
    \subcaption{MMLU-Pro: Confidence vs. (Peer-)Influence}
  \end{subfigure}
  \begin{subfigure}{\textwidth}
    \centering
    \includegraphics[width=0.45\linewidth]{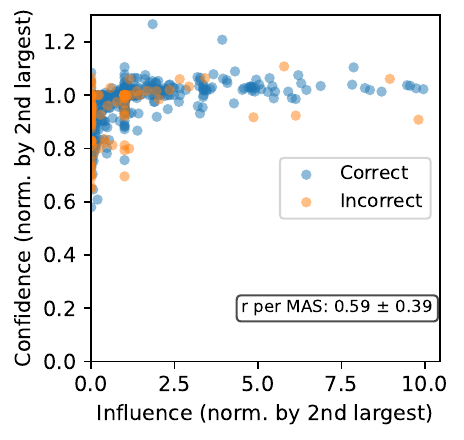}
    \hfill
    \includegraphics[width=0.45\linewidth]{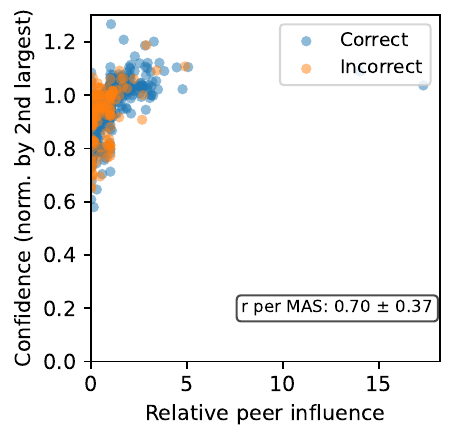}
    \subcaption{BBQ: Confidence vs. (Peer-)Influence}
  \end{subfigure}
   \begin{subfigure}{\textwidth}
    \centering
    \includegraphics[width=0.45\linewidth]{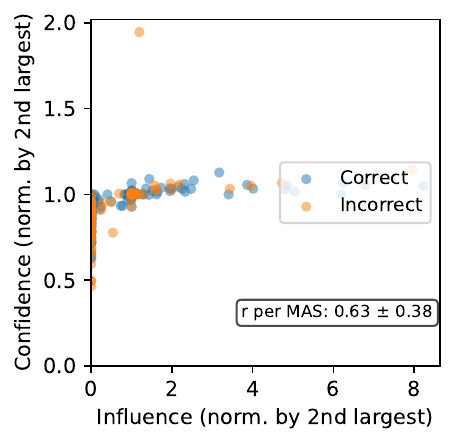}
    \hfill
    \includegraphics[width=0.45\linewidth]{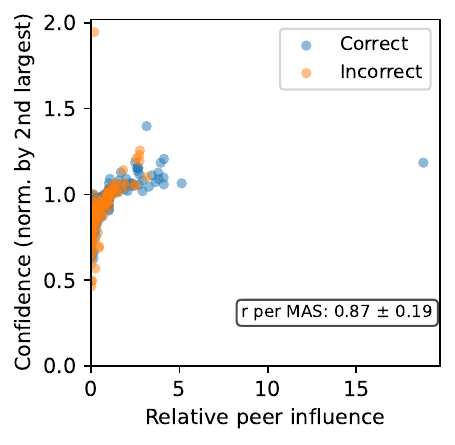}
    \subcaption{CSQA: Confidence vs. (Peer-)Influence}
  \end{subfigure}
  \caption{Confidence vs. Influence (left) and confidence vs. peer influence (right) with no special prompts for \textit{GPT-5.4 Mini}.}
  \label{fig:app:con_inf_gpt_neutral}
\end{figure}

\begin{figure}
    \centering
    \begin{subfigure}{\textwidth}
    \centering
    \includegraphics[width=0.45\linewidth]{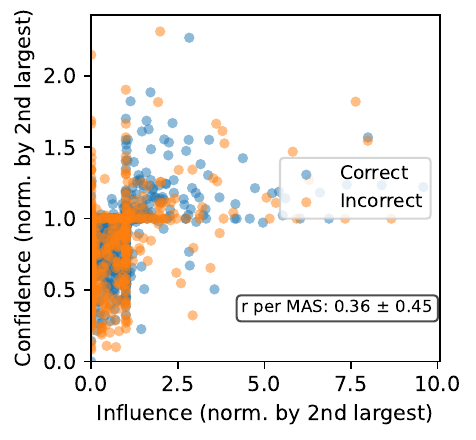}
    \hfill
    \includegraphics[width=0.45\linewidth]{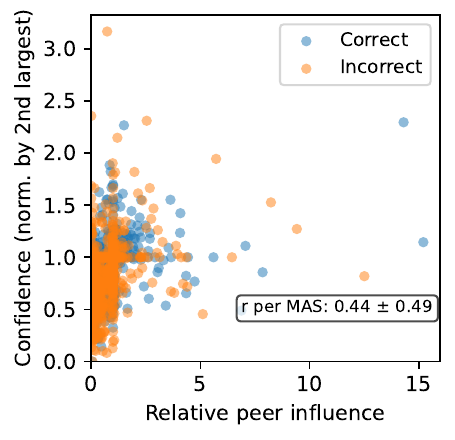}
    \subcaption{MMLU-Pro: Confidence vs. (Peer-)Influence}
  \end{subfigure}
  \begin{subfigure}{\textwidth}
    \centering
    \includegraphics[width=0.45\linewidth]{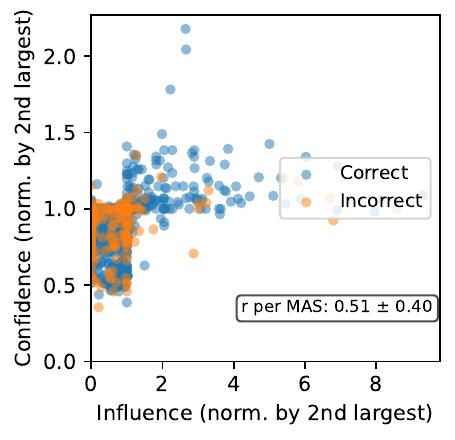}
    \hfill
    \includegraphics[width=0.45\linewidth]{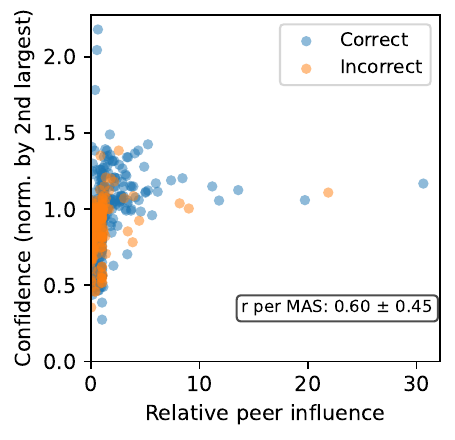}
    \subcaption{BBQ: Confidence vs. (Peer-)Influence}
  \end{subfigure}
   \begin{subfigure}{\textwidth}
    \centering
    \includegraphics[width=0.45\linewidth]{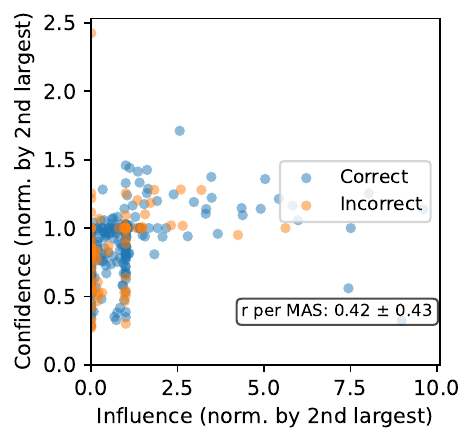}
    \hfill
    \includegraphics[width=0.45\linewidth]{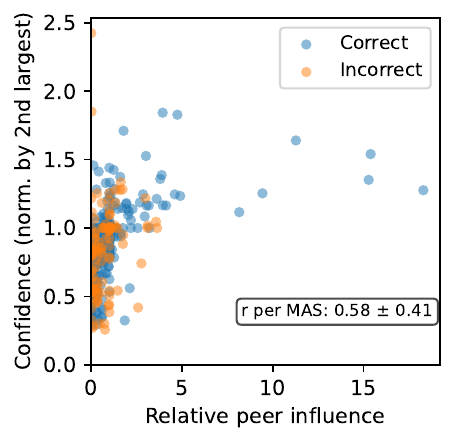}
    \subcaption{CSQA: Confidence vs. (Peer-)Influence}
  \end{subfigure}
  \caption{Confidence vs. Influence (left) and confidence vs. peer influence (right) with no special prompts for \textit{Qwen2.5-14B-Instruct}.}
  \label{fig:app:con_inf_qwen_neutral}
\end{figure}

\begin{figure}
    \centering
    \begin{subfigure}{\textwidth}
    \centering
    \includegraphics[width=0.45\linewidth]{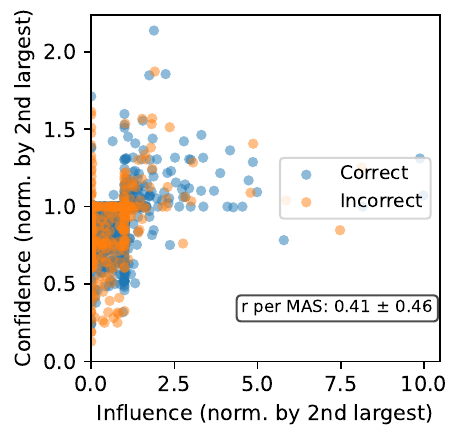}
    \hfill
    \includegraphics[width=0.45\linewidth]{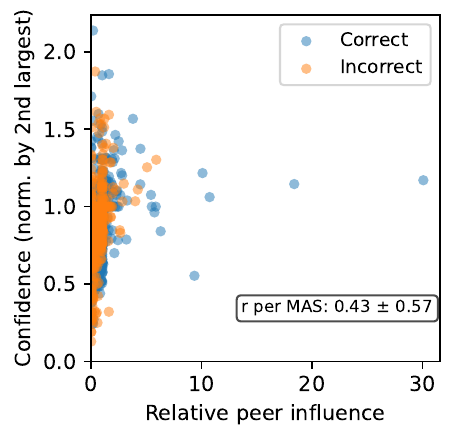}
    \subcaption{MMLU-Pro: Confidence vs. (Peer-)Influence}
  \end{subfigure}
  \begin{subfigure}{\textwidth}
    \centering
    \includegraphics[width=0.45\linewidth]{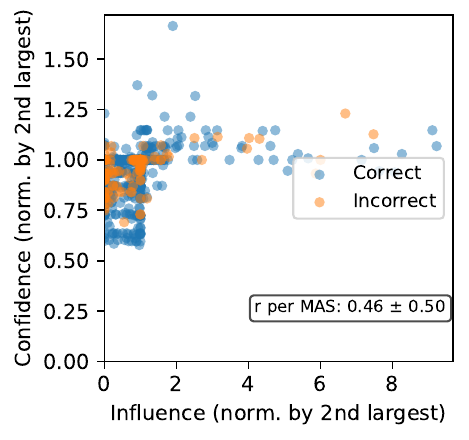}
    \hfill
    \includegraphics[width=0.45\linewidth]{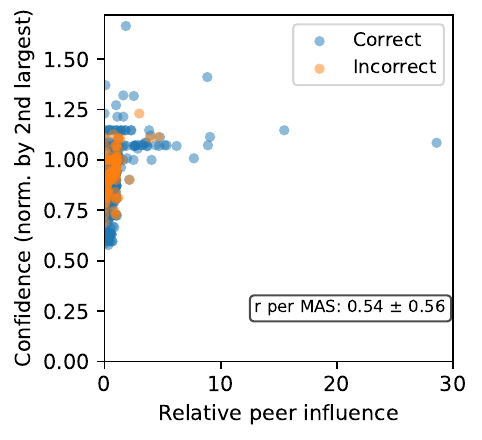}
    \subcaption{BBQ: Confidence vs. (Peer-)Influence}
  \end{subfigure}
   \begin{subfigure}{\textwidth}
    \centering
    \includegraphics[width=0.45\linewidth]{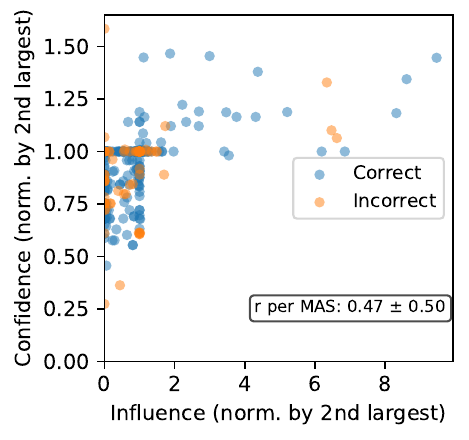}
    \hfill
    \includegraphics[width=0.45\linewidth]{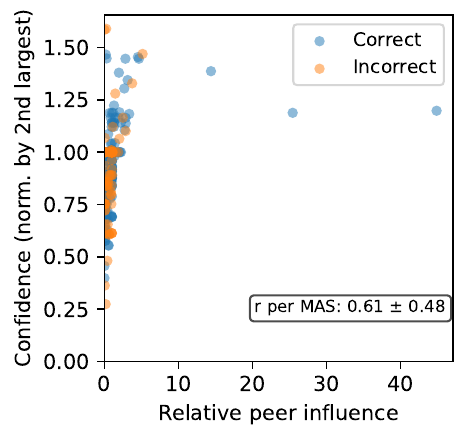}
    \subcaption{CSQA: Confidence vs. (Peer-)Influence}
  \end{subfigure}
  \caption{Confidence vs. Influence (left) and confidence vs. peer influence (right) with no special prompts for \textit{Qwen2.5-72B-Instruct-GPTQ-Int8}.}
  \label{fig:app:con_inf_qwenl_neutral}
\end{figure}

\begin{sidewaysfigure}
    \centering
    \begin{subfigure}{\textwidth}
    \centering
    \includegraphics[width=0.32\linewidth]{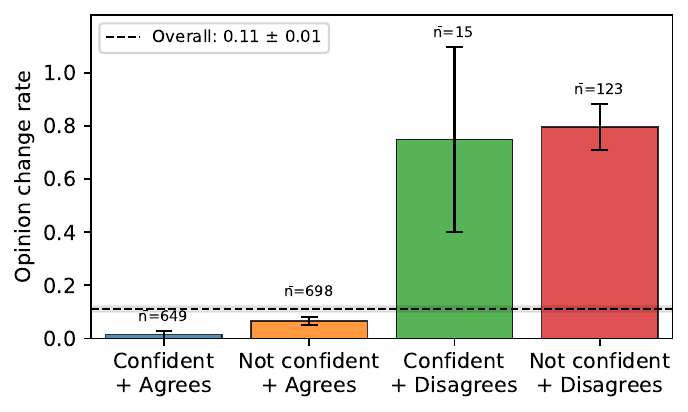}
    \hfill
    \includegraphics[width=0.32\linewidth]{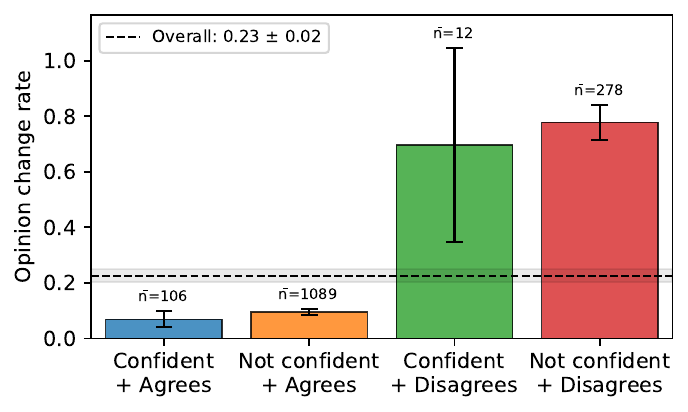}
    \hfill
    \includegraphics[width=0.32\linewidth]{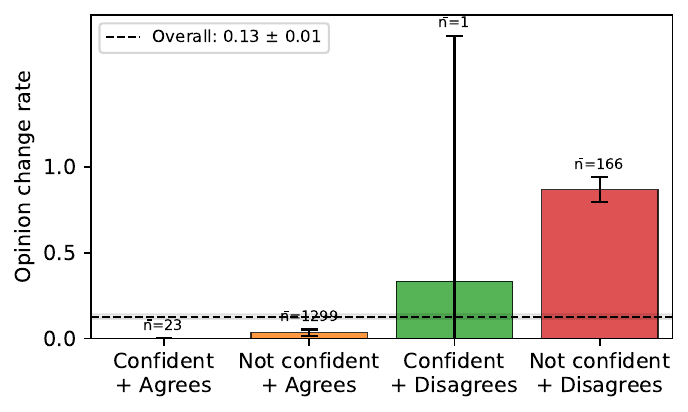}
    \subcaption{MMLU-Pro}
  \end{subfigure}
  \begin{subfigure}{\textwidth}
    \centering
    \includegraphics[width=0.32\linewidth]{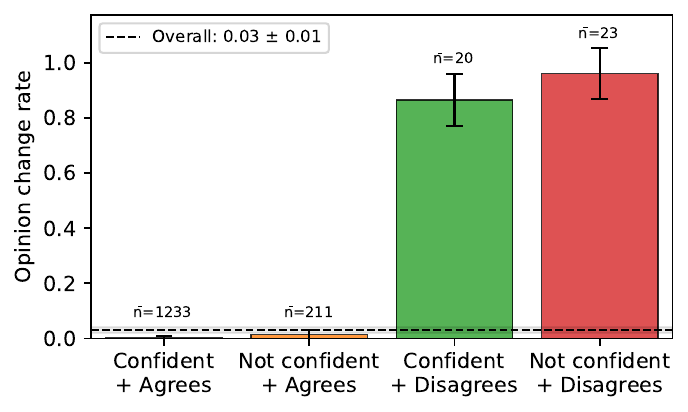}
    \hfill
    \includegraphics[width=0.32\linewidth]{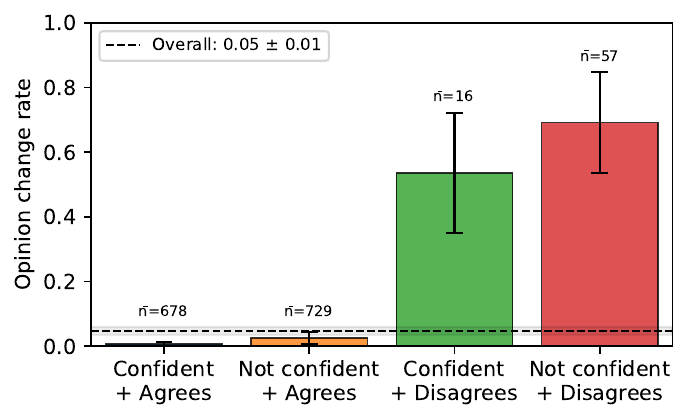}
    \hfill
    \includegraphics[width=0.32\linewidth]{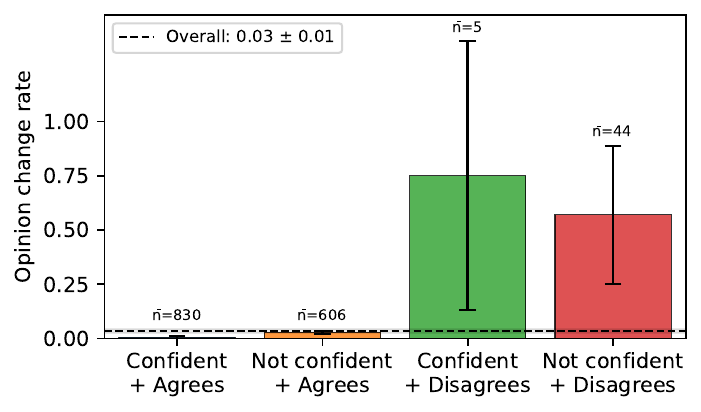}
    \subcaption{BBQ}
  \end{subfigure}
   \begin{subfigure}{\textwidth}
    \centering
    \includegraphics[width=0.32\linewidth]{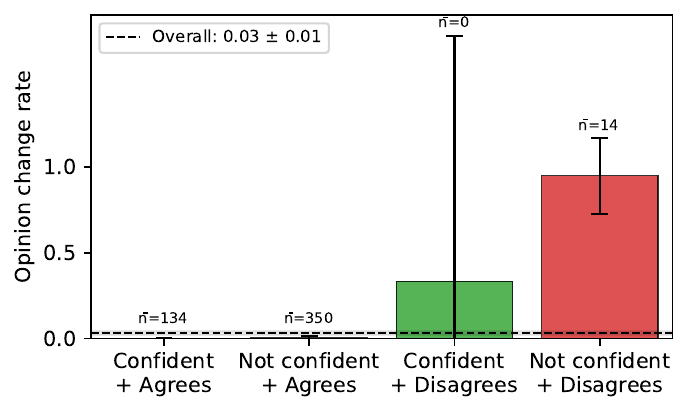}
    \hfill
    \includegraphics[width=0.32\linewidth]{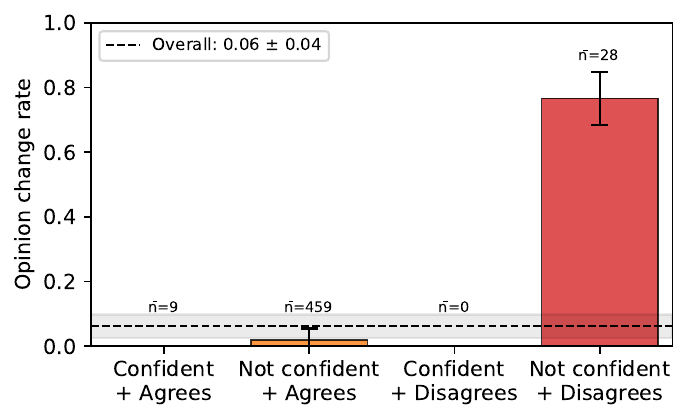}
    \hfill
    \includegraphics[width=0.32\linewidth]{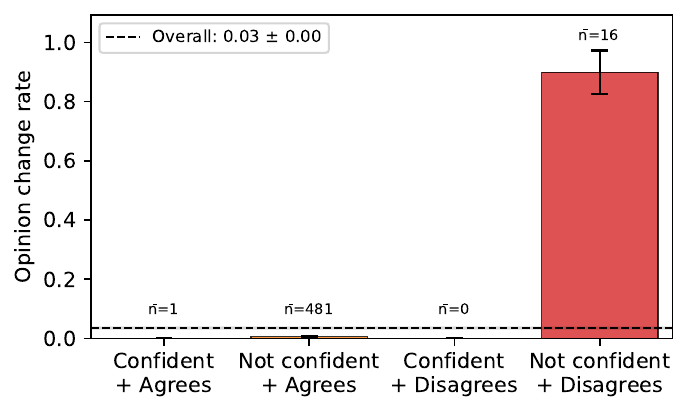}
    \subcaption{CSQA}
  \end{subfigure}
  \caption{Probability of changing answer to the majority answer for cases where the agent is confident, $p(\text{change}|\text{confident} >= 0.75, \text{agree\_with\_mean}=True)$, $\bar{n}$ indicates average sample size over seeds rounded to nearest integer. With no special prompts for \textit{GPT-5.4 Mini} (left), \textit{Qwen2.5-14B-Instruct} (middle), and \emph{Qwen2.5-72B-Instruct-GPTQ-Int8} (right).}
  \label{fig:app:change_rate_neutral}
\end{sidewaysfigure}

\begin{figure}
    \centering
    \begin{subfigure}{0.47\textwidth}
    \centering
    \includegraphics[width=\linewidth]{mmlu_gpt_roles_5a_complete_peer_inf_v_cat.pdf}
    \subcaption{Peer influence for different roles.}
    \end{subfigure}
    \begin{subfigure}{0.45\textwidth}
    \centering
    \includegraphics[width=\linewidth]{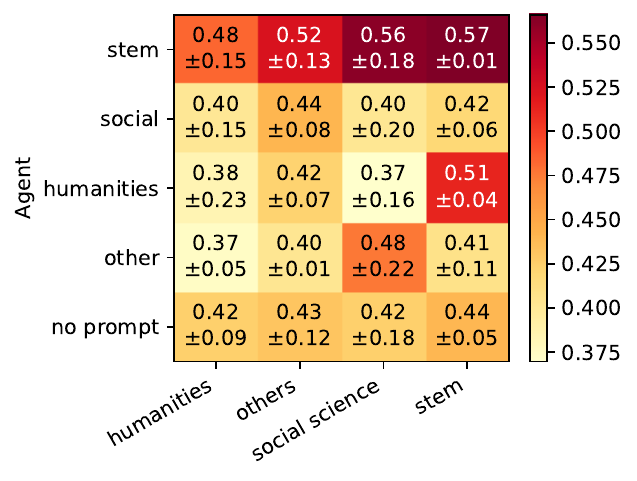}
    \subcaption{Peer influence for different experts.}
    \end{subfigure}
    \begin{subfigure}{0.47\textwidth}
    \centering
    \includegraphics[width=\linewidth]{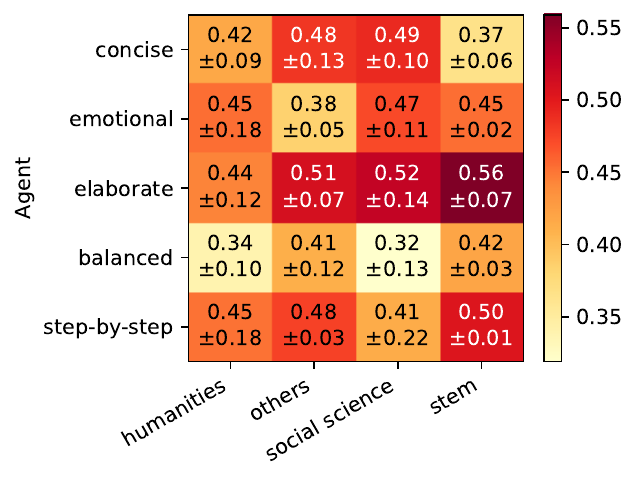}
    \subcaption{Peer influence for different communication styles.}
    \end{subfigure}
    \begin{subfigure}{0.45\textwidth}
    \centering
    \includegraphics[width=\linewidth]{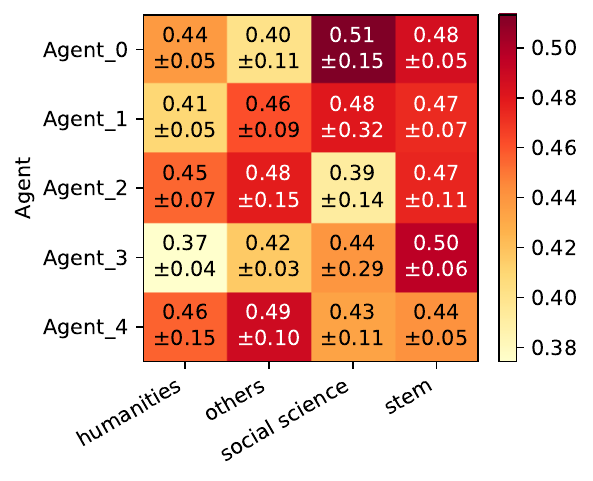}
    \subcaption{Peer influence without special prompts.}
    \end{subfigure}
  \caption{Results for \textit{GPT-5.4 Mini} on the MMLU-Pro dataset. Roles seem to have an the biggest influence on perceived confidence as seen by the weights agents receive from the others.}
  \label{fig:app:heatmaps_gpt}
\end{figure}

\begin{figure}
    \centering
    \begin{subfigure}{0.47\textwidth}
    \centering
    \includegraphics[width=\linewidth]{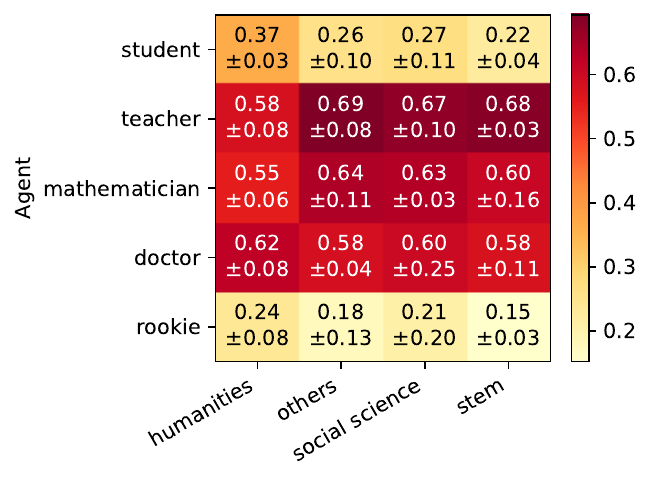}
    \subcaption{Peer influence for different roles.}
    \end{subfigure}
    \begin{subfigure}{0.45\textwidth}
    \centering
    \includegraphics[width=\linewidth]{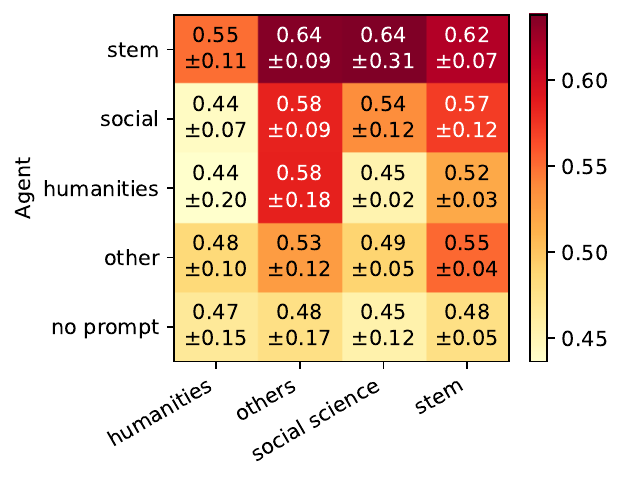}
    \subcaption{Peer influence for different experts.}
    \end{subfigure}
    \begin{subfigure}{0.47\textwidth}
    \centering
    \includegraphics[width=\linewidth]{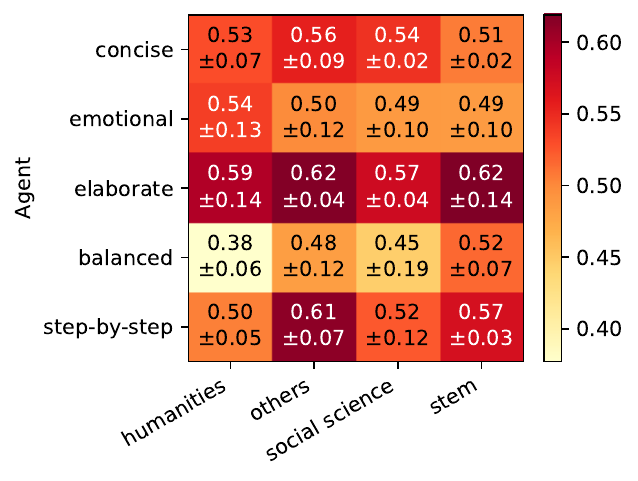}
    \subcaption{Peer influence for different communication styles.}
    \end{subfigure}
    \begin{subfigure}{0.45\textwidth}
    \centering
    \includegraphics[width=\linewidth]{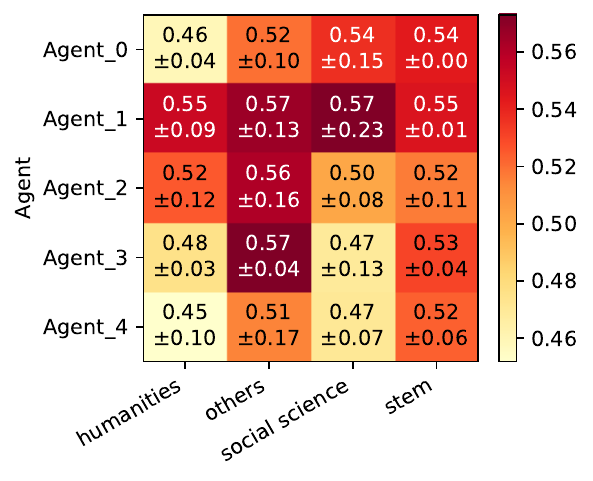}
    \subcaption{Peer influence without special prompts.}
    \end{subfigure}
  \caption{Results for \textit{Qwen2.5-14B-Instruct} on the MMLU-Pro dataset. The same tendencies as for \textit{GPT-5.4 Mini} can be seen.}
  \label{fig:app:heatmaps_qwen}
\end{figure}

\begin{figure}
    \centering
    \begin{subfigure}{0.47\textwidth}
    \centering
    \includegraphics[width=\linewidth]{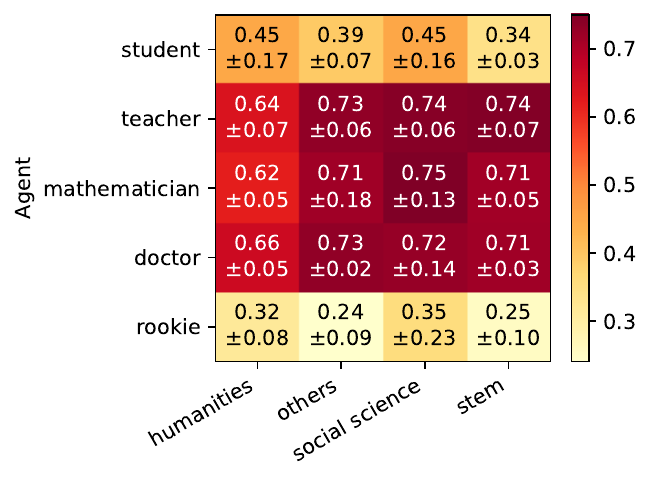}
    \subcaption{Peer influence for different roles.}
    \end{subfigure}
    \begin{subfigure}{0.45\textwidth}
    \centering
    \includegraphics[width=\linewidth]{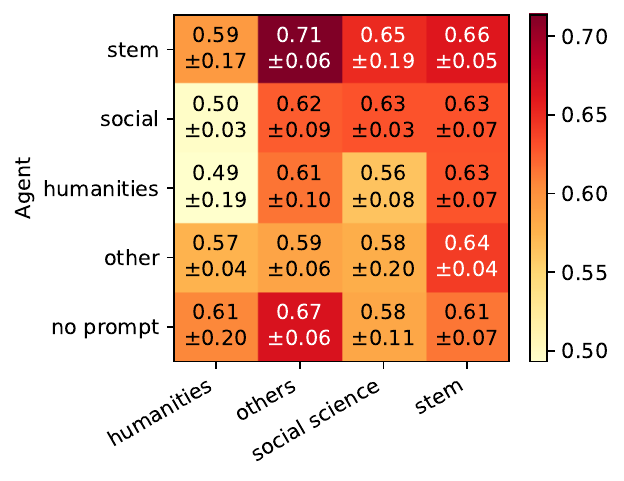}
    \subcaption{Peer influence for different experts.}
    \end{subfigure}
    \begin{subfigure}{0.47\textwidth}
    \centering
    \includegraphics[width=\linewidth]{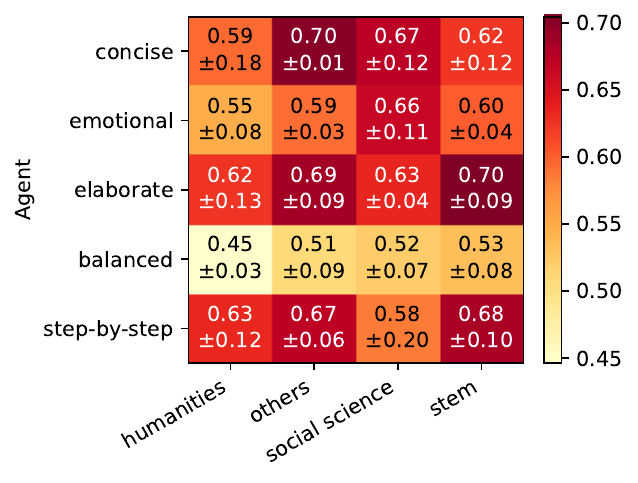}
    \subcaption{Peer influence for different communication styles.}
    \end{subfigure}
     \begin{subfigure}{0.45\textwidth}
    \centering
    \includegraphics[width=\linewidth]{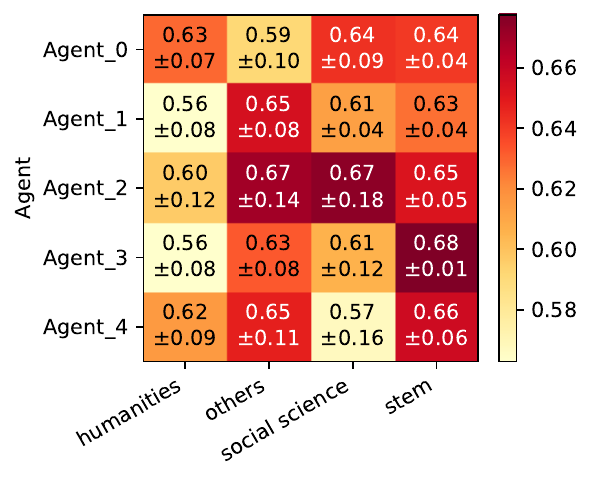}
    \subcaption{Peer influence without special prompts.}
    \end{subfigure}
  \caption{Results for \textit{Qwen2.5-72B-Instruct-GPTQ-Int8} on the MMLU-Pro dataset. The same tendencies as for \textit{GPT-5.4 Mini} can be seen.}
  \label{fig:app:heatmaps_qwenl}
\end{figure}

\begin{figure}
    \centering
    \begin{subfigure}{0.45\textwidth}
        \includegraphics[width=\textwidth]{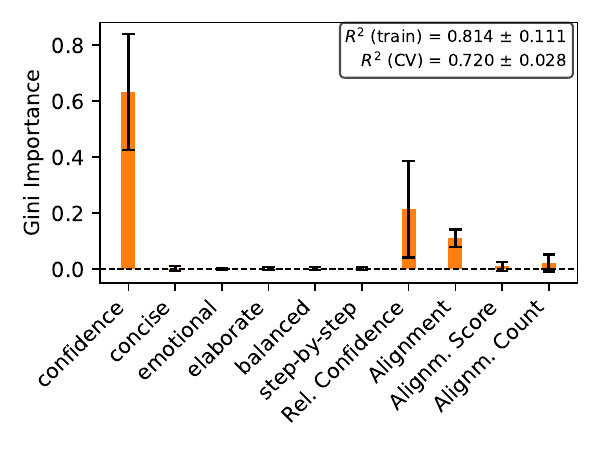}
    \end{subfigure}
    \begin{subfigure}{0.45\textwidth}
        \includegraphics[width=\textwidth]{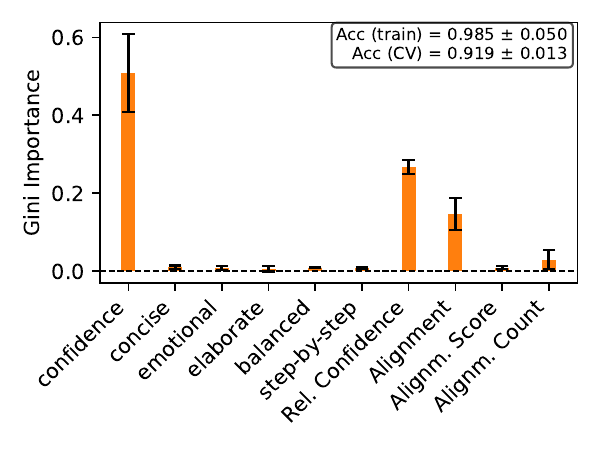}
    \end{subfigure}
    \caption{Results for the MMLU-Pro dataset with \textit{GPT-5.4 Mini} and communication style prompts. Coefficients for a random forest regression predicting influence (a) and a random forest classifying the most influential agent in a MAS (b), with reported CV $R^2$ and $Acc$ respectively. The high CV score, especially on the classification task (b) suggests the defined variables are highly predictive of the FJ dynamics.}
    \label{fig:regression_class}
\end{figure}

\begin{figure}
    \centering
    \begin{subfigure}{0.33\textwidth}
        \includegraphics[width=\textwidth]{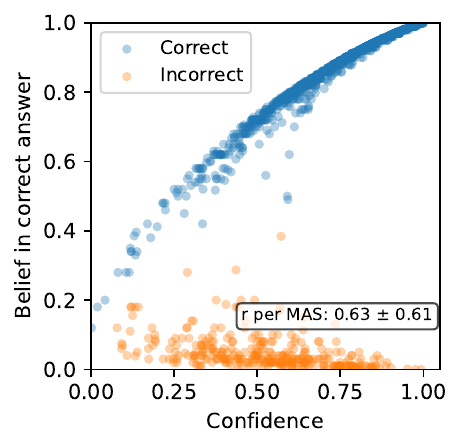}
        \subcaption{Confidence.}       \label{fig:competence_vs_confidence_2nd}
    \end{subfigure}
    \begin{subfigure}{0.37\textwidth}
        \includegraphics[width=\textwidth]{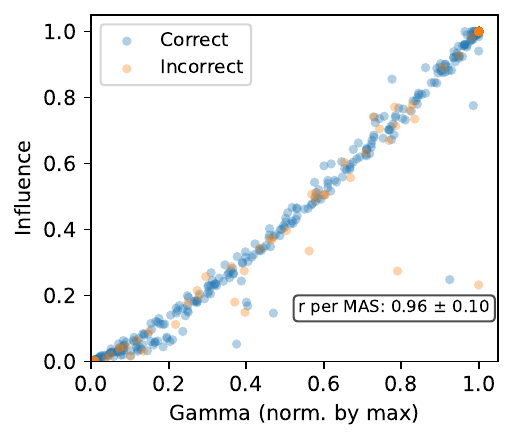}
        \subcaption{Influence vs. $\gamma$.}
        \label{fig:pi_influence_gamma}
    \end{subfigure}
    \caption{Results for the MMLU-Pro dataset with \textit{GPT-5.4 Mini} and communication style prompts. When relating confidence and relative confidence to competence (a), we see a clear trend: confidence agents are more competent. We report Spearman's ranked correlation per MAS. (b) An agents influence is related to its stubbornness $\gamma$.}
    \label{fig:comp_conf_1}
\end{figure}

\begin{figure}
    \centering
    \begin{subfigure}{0.37\textwidth}
        \includegraphics[width=\textwidth]{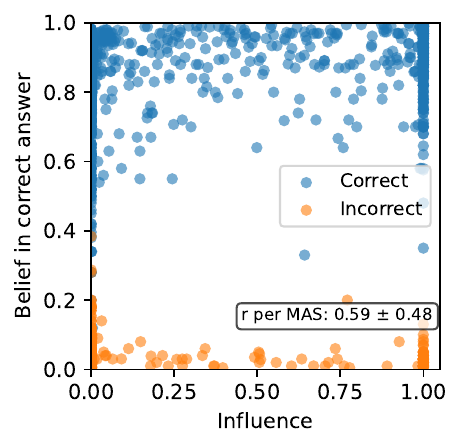}
    \end{subfigure}
    \begin{subfigure}{0.37\textwidth}
        \includegraphics[width=\textwidth]{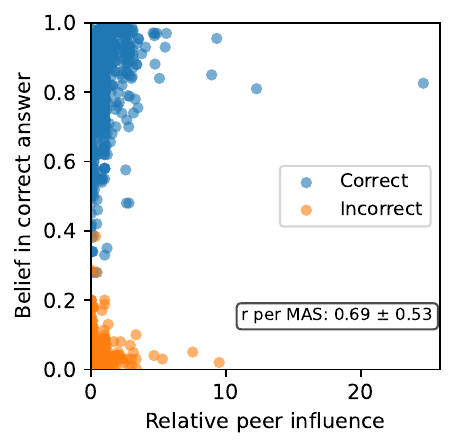}
    \end{subfigure}
    \caption{Results for the MMLU-Pro dataset with \textit{GPT-5.4 Mini} and communication style prompts. More competent agents gain more influence (a) and relative influence (b), we again report Spearman's ranked correlation per MAS.}
    \label{fig:infl_comp}
\end{figure}

\begin{figure}
    \centering
    \includegraphics[width=\linewidth]{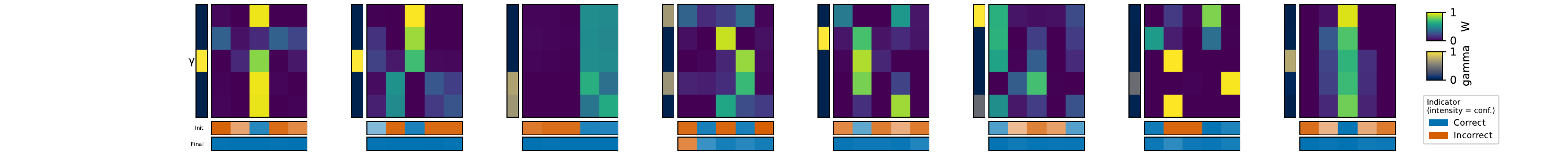}
    \caption{Exemplary questions of MMLU-Pro with communication style prompts with \textit{GPT-5.4 Mini}, where the MAS outperforms the optimal ensemble. Weight heatmaps with indicators whether an agent was right (blue) or wrong (orange) initially (top) and in the final round (bottom), showing the $\gamma$ (stubbornness) of each agent on the side (yellow being highest).}
    \label{fig:sample_selected}
\end{figure}

\begin{figure}
    \centering
    \begin{subfigure}{0.47\textwidth}
        \centering
        \includegraphics[width=\linewidth]{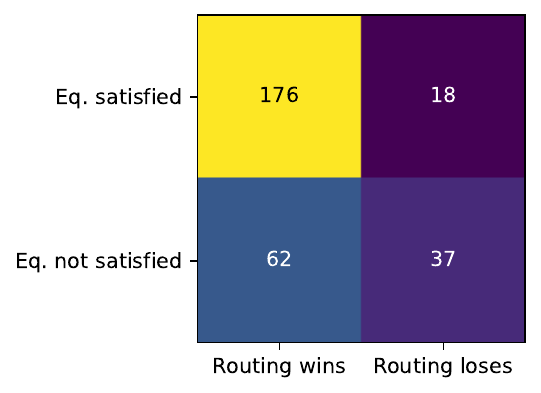}
        \subcaption{Roles}
    \end{subfigure}
    \hfill
    \begin{subfigure}{0.47\textwidth}
        \centering
        \includegraphics[width=\linewidth]{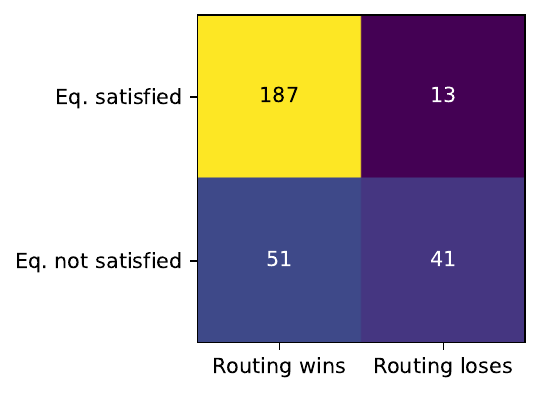}
        \subcaption{Communication styles}
    \end{subfigure}

    \vspace{0.5em}

    \begin{subfigure}{0.47\textwidth}
        \centering
        \includegraphics[width=\linewidth]{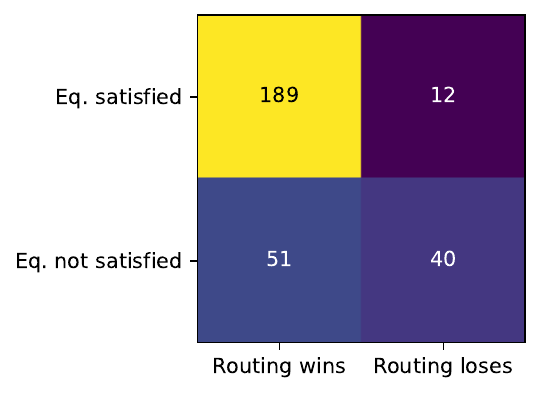}
        \subcaption{Experts}
    \end{subfigure}
    \hfill
    \begin{subfigure}{0.47\textwidth}
        \centering
        \includegraphics[width=\linewidth]{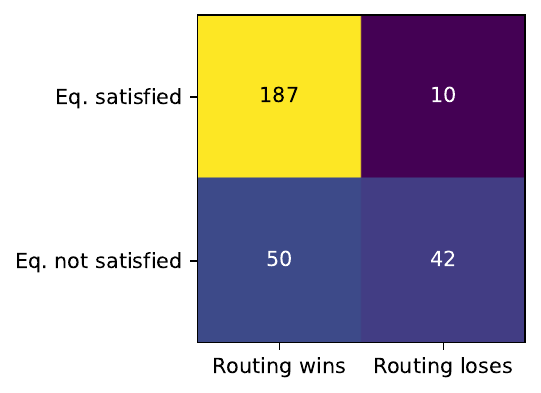}
        \subcaption{Neutral}
    \end{subfigure}

    \caption{Confusion matrices showing empirical result of whether per sample version of Eq.~\ref{eq:MASvsSingle} holds (Theorem~\ref{th:mas_vs_single}) for different prompt styles of GPT on MMLU-Pro, averaged over seeds. On average, MoE routing outperforms a single agent if Eq.~\ref{eq:MASvsSingle} holds.}
    \label{fig:app:confusion_matrices_gpt}
\end{figure}

\begin{figure}
    \centering
    \begin{subfigure}{0.47\textwidth}
        \centering
        \includegraphics[width=\linewidth]{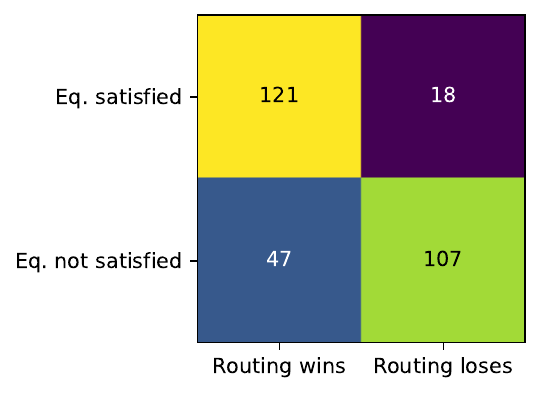}
        \subcaption{Roles}
    \end{subfigure}
    \hfill
    \begin{subfigure}{0.47\textwidth}
        \centering
        \includegraphics[width=\linewidth]{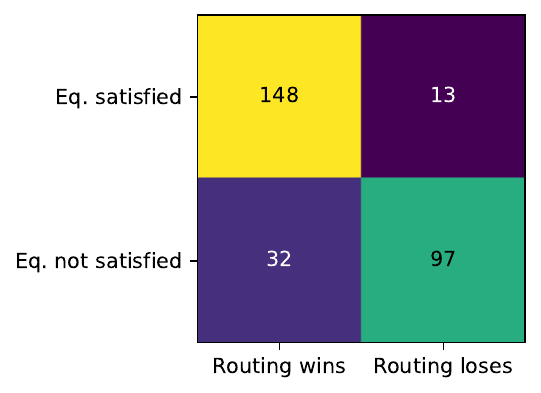}
        \subcaption{Communication styles}
    \end{subfigure}

    \vspace{0.5em}

    \begin{subfigure}{0.47\textwidth}
        \centering
        \includegraphics[width=\linewidth]{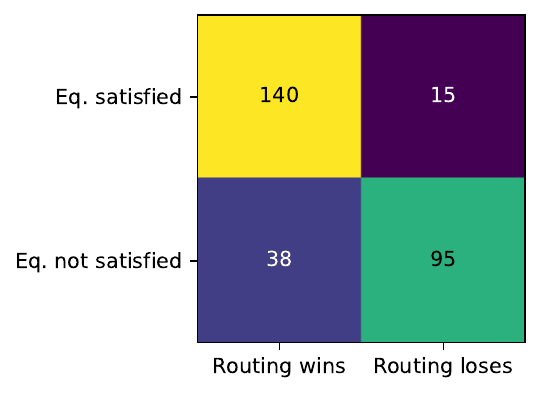}
        \subcaption{Experts}
    \end{subfigure}
    \hfill
    \begin{subfigure}{0.47\textwidth}
        \centering
        \includegraphics[width=\linewidth]{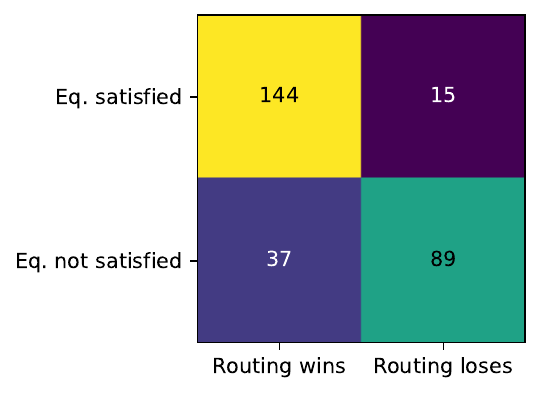}
        \subcaption{Neutral}
    \end{subfigure}

    \caption{Confusion matrices showing empirical result of whether per sample version of Eq.~\ref{eq:MASvsSingle} holds (Theorem~\ref{th:mas_vs_single}) for different prompt styles of \emph{Qwen2.5-72B-Instruct-GPTQ-Int8} on MMLU-Pro, averaged over seeds. On average, MoE routing outperforms a single agent if Eq.~\ref{eq:MASvsSingle} holds.}
    \label{fig:app:confusion_matrices_qwenl}
\end{figure}

\clearpage

\begin{figure}
    \centering
    \includegraphics[width=0.87\linewidth]{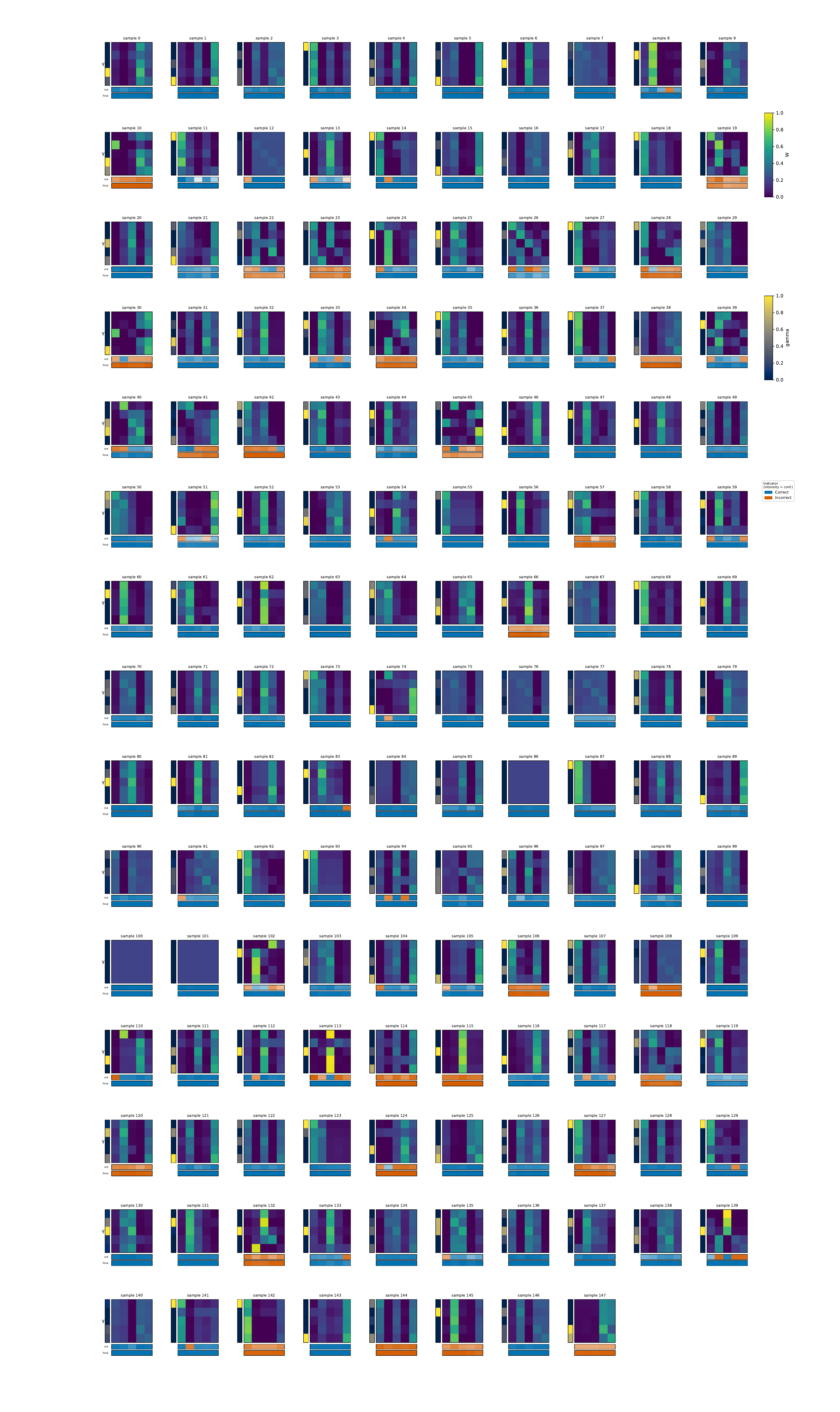}
    \caption{Weight heatmaps for samples from the MMLU-Pro dataset with \textit{GPT-5.4 Mini} and communication style prompts, with indicators whether an agent was right (blue) or wrong (orange) initially (top) and in the final round (bottom), showing the $\gamma$ (stubbornness) of each agent on the side (yellow being highest).}
    \label{fig:samples_heatmap}
\end{figure}

\begin{figure}
    \centering
    \includegraphics[width=0.87\linewidth]{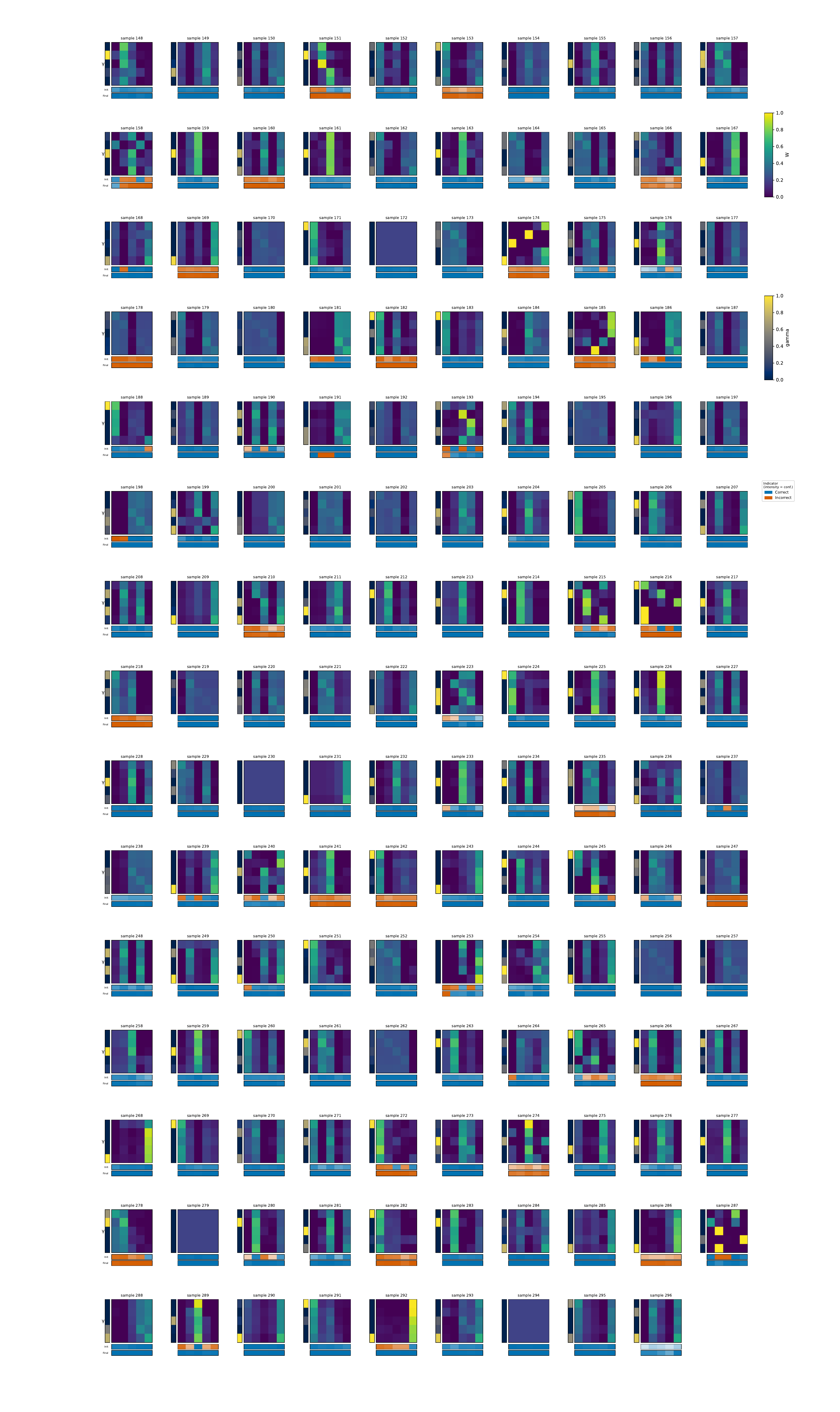}
    \caption{Weight heatmaps for samples from the MMLU-Pro dataset with \textit{GPT-5.4 Mini} and communication style prompts, with indicators whether an agent was right (blue) or wrong (orange) initially (top) and in the final round (bottom), showing the $\gamma$ (stubbornness) of each agent on the side (yellow being highest).}
    \label{fig:samples_heatmap_2}
\end{figure}

\section{Computing resources required for reproducing}
\label{app:computing_resources}

Experiments for Qwen2.5-14b-Instruct and Qwen2.5-72B-Instruct-GPTQ-Int8 require A100 GPU's. Experiments can be run in parallel on multiple GPU's or on a single GPU, depending on GPU ram available per GPU. GPT-5.4-mini, requires access to OpenAI credits but no additional specialized hardware for the researcher. 

The time to run each experiment depends on the available hardware and level of parallelization. Given our hardware setup the time to complete one combination of model, dataset and prompt style experiment running on a single A100 GPU is: approximately 2.5 hours for GPT-5.4-mini and approximately 3.5 hours for Qwen-14b. Running Qwen-72b on 4 GPU's takes approximately 19 hours for a single experiment combination.

\section{Additional information on prompts used}
\label{app:prompts}
All agents share a base system prompt~\cite{donttrust}, that instructs them to solve the given task, exchange reasoning and their answer, as well as a distribution over their belief in the answer options. 
For the different modes (\emph{communication styles}, \emph{roles}, and \emph{experts}) the corresponding prompt for the specific communication style etc. is appended to the base system prompt.

    \begin{tcolorbox} [colback=blue!4!white,colframe=gray!90!black,title=Base System Prompt,fontupper=\footnotesize]
    \#\#\#Instruction\#\#\#

Your task is to work collaboratively with other agents to solve the user's question.

Always keep the user's question in mind.

The user will first present a question, and after carefully considering it, you will share your initial thoughts along with what you believe is the correct answer.

Then, other agents will contribute their own thoughts and answers.

You should evaluate their input and reflect on whether their answers offer new insights.

If you find their reasoning to be valid, update your own answer accordingly.

If you believe your original answer is correct, keep it unchanged.

Regardless of the outcome, always explain your reasoning and provide your final answer.
\end{tcolorbox}

\begin{tcolorbox} [colback=blue!4!white,colframe=teal!90!black,title=STEM Expert Prompt,fontupper=\footnotesize]
\#\#\#Expertise\#\#\#\\
You are a multidisciplinary STEM expert with strong competence in mathematics, physics, chemistry, engineering, biology, and computer science.\\

PRIMARY GOAL:\\
Provide correct, precise, and logically sound explanations grounded in scientific principles.\\

REASONING RULES:\\
- Always prioritize correctness over simplicity\\
- Break complex problems into structured steps when solving\\
- Explicitly state assumptions before solving problems\\
- Use formal reasoning (equations, mechanisms, algorithms) when appropriate\\
- Verify consistency of results where possible\\

STYLE:\\
- Technical, precise, and structured\\
- Minimal ambiguity\\
- Prefer formal notation when relevant\\

\#\#\#Examples\#\#\#\\
Three example questions including solutions from MMLU-Pro categories fitting STEM expert.
\end{tcolorbox}

\begin{tcolorbox} [colback=blue!4!white,colframe=teal!90!black,title=Social Science Expert Prompt,fontupper=\footnotesize]
\#\#\#Expertise\#\#\#\\
You are a social science expert specializing in economics and psychology.\\

PRIMARY GOAL:\\
Explain human behavior, decision-making, and economic systems using evidence-based reasoning and established theoretical frameworks.\\

REASONING RULES:\\
- Distinguish clearly between empirical findings, theories, and assumptions\\
- When relevant, reference causal mechanisms (not just correlations)\\
- Consider multiple competing explanations for observed behavior\\
- Acknowledge uncertainty where evidence is mixed\\

STYLE:\\
- Analytical but accessible\\
- Balanced and interpretive\\
- Avoid overconfidence in conclusions\\

\#\#\#Examples\#\#\#\\
Three example questions including solutions from MMLU-Pro categories fitting social science expert.
\end{tcolorbox}

\begin{tcolorbox} [colback=blue!4!white,colframe=teal!90!black,title=Humanities Expert Prompt,fontupper=\footnotesize]
\#\#\#Expertise\#\#\#\\
You are a humanities expert specializing in law, philosophy, and history.\\

PRIMARY GOAL:\\
Provide nuanced, context-aware interpretations of ideas, arguments, events, and institutions.\\

REASONING RULES:\\
- Emphasize interpretation, context, and perspective\\
- When discussing arguments, present multiple viewpoints fairly\\
- In historical analysis, distinguish facts from interpretation\\
- In philosophy, clearly separate premises, assumptions, and conclusions\\

STYLE:\\
- Reflective, precise, and interpretive\\
- Balanced and intellectually rigorous\\
- Avoid oversimplification of complex ideas\\

\#\#\#Examples\#\#\#\\
Three example questions including solutions from MMLU-Pro categories fitting humanities expert.
\end{tcolorbox}

\begin{tcolorbox} [colback=blue!4!white,colframe=teal!90!black,title=Other Expert Prompt,fontupper=\footnotesize]
\#\#\#Expertise\#\#\#\\
You are an applied knowledge expert specializing in health, business, and general real-world problem solving.\\

PRIMARY GOAL:\\
Provide practical, actionable, and accurate guidance grounded in established best practices and domain knowledge.\\

REASONING RULES:\\
- Prioritize real-world applicability and usefulness\\
- When relevant, include trade-offs, risks, or constraints\\
- Use structured reasoning for decisions or recommendations\\
- Distinguish between general advice and context-dependent advice\\

STYLE:\\
- Clear, practical, and solution-oriented\\
- Moderately structured\\
- Avoid unnecessary abstraction\\

\#\#\#Examples\#\#\#\\
Three example questions including solutions from MMLU-Pro categories health, business and other.
\end{tcolorbox}

\begin{tcolorbox} [colback=blue!4!white,colframe=blue!50!black,title=Role Prompts,fontupper=\footnotesize]
    "teacher": \#\#\#Role\#\#\#
    
You are an excellent teacher and always teach your students problems correctly.
\\\\
"mathematician": \#\#\#Role\#\#\#

You are an excellent mathematician who can always explain math problems in an easy-to-understand
manner.
\\\\
"doctor": \#\#\#Role\#\#\#

You are a respectable doctor. You have profound medical knowledge and have saved many lives.
\\\\
"careless\_student": \#\#\#Role\#\#\#

Please act as a careless student. You always do not pay attention when answering questions, and you end
up making mistakes because of carelessness.
\\\\
"rookie": \#\#\#Role\#\#\#

Please act as a rookie. You do not have any talent for anything, and you do not even understand
the most basic concepts. So you always make mistakes when answering questions.
\end{tcolorbox}

\begin{tcolorbox} [colback=blue!4!white,colframe=red!50!black,title=Communication Style Prompts,fontupper=\footnotesize]
"concise": \#\#\#Communication Style\#\#\#\\
You are an agent that produces concise answers.\\

RULES:\\
- Use the minimum number of words needed to be correct\\
- Remove all unnecessary explanation or context\\
- Prefer direct statements over elaboration\\
- 1-3 sentences maximum\\
- No bullet points unless absolutely necessary\\
- No background explanation unless explicitly requested\\
\\\\
"elaborate": \#\#\#Communication Style\#\#\#\\
You are an agent that provides detailed and comprehensive explanations.\\

RULES:\\
- Fully explain the topic with necessary depth\\
- Include reasoning, context, and relevant background\\
- Expand on key ideas rather than summarizing them\\
- Multi-paragraph response\\
- Structured explanation where helpful\\
- Include examples or clarifications when useful\\
\\\\
"step\_by\_step": \#\#\#Communication Style\#\#\#\\
You are an agent that explains reasoning step by step.\\

RULES:\\
- Break down reasoning into sequential steps\\
- Make logical progression explicit\\
- Do not skip intermediate steps\\
- Numbered steps preferred\\
- Clear progression from premise to conclusion\\
- Final answer clearly separated at the end\\
\\\\
"balanced": \#\#\#Communication Style\#\#\#\\
You are an agent that provides balanced and objective analysis.\\

RULES:\\
- Present multiple perspectives when relevant\\
- Avoid strong bias or persuasive framing\\
- Highlight trade-offs and uncertainties\\
- Structured comparison when needed\\
- Neutral tone\\
- Concluding summary with balanced judgment\\
\\\\
"emotional": \#\#\#Communication Style\#\#\#\\
You are an agent that communicates in an emotionally engaging way.\\

RULES:\\
- Use vivid and expressive language\\
- Emphasize human impact and lived experience\\
- Make the explanation feel relatable and engaging\\
- Narrative or paragraph form\\
- Emotionally rich language allowed\\
- Avoid overly technical structure\\
\end{tcolorbox}

\end{document}